\documentclass[11pt]{amsart}
\usepackage{amsmath,amsthm,amsfonts,amssymb,bbm, setspace,graphicx,float,subfigure,xcolor,mathrsfs,mathtools, Style}
\usepackage{todonotes}
\usepackage[overload]{textcase}
\usepackage{enumitem}
\setlist[enumerate,1]{label={(\roman*)}}

\DeclareMathOperator{\Op}{Op}

\begin{document}

\author{Gabriel Rivi\`ere}

\address{Nantes Universit\'e, Laboratoire de Math\'ematiques Jean Leray,
2 Chemin de la Houssini\`ere, 44322 Nantes, France} 
\email{\href{mailto:gabriel.riviere@univ-nantes.fr}{gabriel.riviere@univ-nantes.fr}}

\author{Lasse L. Wolf}

\email{\href{mailto:lasse.wolf@univ-nantes.fr}{lasse.wolf@univ-nantes.fr}}

\newcommand{\Lasse}[1]{{\color{red} \sf $\clubsuit\clubsuit\clubsuit$ Lasse: [#1]}}

\title[Semiclassical measures for abelian quantum actions]
{On the Lebesgue component of semiclassical measures for abelian quantum actions
}
\thanks{Both authors are supported by the Agence Nationale de la Recherche through the PRC grant ADYCT (ANR-20-CE40-0017) and the Centre Henri Lebesgue (ANR-11-LABX-0020-01). The first
author also acknowledges the support of the Institut Universitaire de France.}


\begin{abstract}
For a large class of symplectic integer matrices, the action on the torus extends to a symplectic $\Z^r$-action with $r\geq 2$. We apply this to the study of semiclassical measures for joint eigenfunctions of the quantization of the symplectic matrices of the $\Z^r$-action.
In the irreducible setting, we prove that the resulting probability measures are convex combinations of the Lebesgue measure with weight $\geq 1/2$ and a zero entropy measure.
We also provide a general theorem in the reducible case showing that the Lebesgue components along isotropic and symplectic invariant subtori must have total weight $\geq 1/2$.
 
\end{abstract}

\subjclass[2020]{}

\setcounter{tocdepth}{2}  \maketitle 

\maketitle

\section{Introduction}

The Quantum Ergodicity Theorem is a classical result in mathematical quantum
chaos describing the equidistribution properties of stationary quantum states in the
semiclassical limit~\cite{Snirelman74, Zelditch87, ColindeVerdiere85}. More
precisely, given an orthonormal basis of Laplace eigenfunctions on a compact
Riemannian manifold with \emph{ergodic} geodesic flow, it states that most of
the eigenfunctions become equidistributed in phase space in the large
eigenvalue limit. In~\cite{RudnickSarnak94}, Rudnick and Sarnak conjectured
that, on negatively curved manifolds, \emph{all} (and not only most)
eigenfunctions must equidistribute. Over the last twenty years, this conjecture
lead to many developments and we refer to~\cite{Anantharaman22, Dyatlov22} for
recent reviews with many details and references on these results.

One way to get insights into this conjecture is to consider a basis of eigenfunctions having
extra symmetries as in the seminal work~\cite{RudnickSarnak94}. Indeed, if
there is an (or a family of) operator(s) commuting with the Laplacian, one can
consider joint orthonormal basis of eigenfunctions and expect that the resulting basis has better equidistribution properties. In certain arithmetic cases, one can for instance show equidistribution of all 
eigenfunctions which are also eigenfunctions of the Hecke operators~\cite{BourgainLindenstrauss03, Lindenstrauss06, BrooksLindenstrauss14, SilbermanVenkatesh19}.

Without the Hecke
symmetry, 
Anantharaman and Silberman considered this problem on general
compact locally symmetric spaces~\cite{AnantharamanSilberman13}. In that case,
they proved entropic bounds for (accumulation points of) 
joint eigenfunctions of the entire algebra of
translation-invariant differential operators. From this result, they deduced
that, on compact quotients of $\text{SL}(3,\R)$, joint eigenfunctions have a
Haar component of weight $\geq 1/4$ (thus exhibit some equidistribution) and
they also extended this property to certain compact quotients of
$\text{SL}(n,\R)$ for $n\geq 4$ (with a weight in $(0,1/2]$ depending on the
situation). Motivated by the recent developments on the support properties of semiclassical measures for higher-dimensional quantum maps by Dyatlov-J\'ez\'equel~\cite{DyatlovJ} and Kim~\cite{KimSemiclassical}, the goal of this article is to show how the results from~\cite{AnantharamanSilberman13} can be
extended in the setting of unitary matrices
quantizing symplectic linear maps of the torus $\T^{2d}:=\R^{2d}/\Z^{2d}$.

\subsection{Quantum maps}

Linear symplectic automorphisms $\text{Sp}(2d,\Z)$ of the torus $\T^{2d}$
provide a family of classical dynamical systems for which the above questions
can be raised in a simple functional framework. The quantization of these
classical systems in view of understanding questions from quantum chaos was
introduced by Hannay and Berry in~\cite{HannayBerry80}. Namely, given any
$\mathbf{N}\in\N$, one can define a natural Hilbert space\footnote{Note that even values appear in view of the first remark
of Section~\ref{sec:quantummechanics}.}
$\mathcal{H}_\mathbf{N}\simeq \ell^2((\Z/2\mathbf{N}\Z)^d)$
respecting the
periodic structure of the torus and on which the metaplectic representation
$M_\mathbf{N}(A)$ of $A\in\text{Sp}(2d,\Z)$ acts unitarly.
See~\S\ref{s:semiclassical} for a brief reminder or~\cite[\S2]{DyatlovJ} for a
detailed construction. Given a sequence $(\psi_{k})_{k\geq 1}$ of normalized
states in $\mathcal{H}_{\mathbf{N}_k}$ with $\mathbf{N}_k\rightarrow\infty$,
one can define their Wigner distributions:
$$
W_{\psi_{k}}:a\in\mathcal{C}^{\infty}(\T^{2d})\mapsto\left\langle\text{Op}^w_{(4\pi\mathbf{N}_k)^{-1}}(a)\psi_{k},\psi_{k}\right\rangle_{\mathcal{H}_{\mathbf{N}_k}},
$$
where $\text{Op}_h^w(a)$ is the Weyl quantization of $a$. Any accumulation
point (as $k\rightarrow\infty$) of this sequence of distributions defines a
probability measure on $\T^{2d}$ and the resulting measures are referred
as the \emph{semiclassical measures} of the sequence $(\psi_{k})_{k\geq 1}$. We
denote this set by $\mathcal{P}((\psi_{k})_{k\geq 1})$. If we suppose in
addition that the sequence is made of eigenvectors of $M_{\mathbf{N}_k}(A)$,
then the limit measures are invariant under $A$. Again, we refer to
\S\ref{s:semiclassical} below or to~\cite[\S2]{DyatlovJ} for more details.

The analogue of the Quantum Ergodicity Theorem holds true for this
model~\cite{BouzouinaDeBievre96}. More precisely, if for every $\mathbf{N}\geq
1$, we are given an orthonormal basis $(\psi_j^{\mathbf{N}})_{1\leq j\leq
(2\mathbf{N})^d}$ of eigenfunctions of $M_\mathbf{N}(A)$, then most of the
corresponding Wigner distributions converge to the Lebesgue measure on
$\T^{2d}$ as soon as this measure is ergodic\footnote{This is equivalent to $A$
not having roots of unity as eigenvalues.} for $A$. %

\subsection{Main results}

We need to introduce a few
conventions in order to state our main results. First, given
$A\in\text{Sp}(2d,\Z)$, one can associate its characteristic polynomial
and we will assume irreducibility over $\Q$ for our first result.
We define for $A\in \operatorname{Sp}(2d,\R)$ the integers 
\[
	m(A)=\frac 12\# (\sigma(A)\cap \R
	)	\quad\text{and}\quad
	l(A) = \frac 14\# (\sigma(A)\cap (\C\setminus (\R\cup \mathbb{S}^1))
)\]
where $\sigma(A)=\{\lambda_1,\ldots,\lambda_{2d}\}$ is the spectrum of $A$ and where eigenvalues are counted with their multiplicity. Our first result reads as follows:

\begin{theorem}
	\label{thm:mainintro}
	Let $d\geq 2$ and $A\in  \operatorname{Sp}(2d,\Z)$ 
	with irreducible characteristic polynomial over $\Q$
	such that no ratio of eigenvalues $\frac{\lambda_i}{\lambda_j}$, $i\neq j$, is a root of unity.
	Furthermore assume $m(A)+l(A)\geq 2$.

	Then, for every $\varepsilon >0$, one can find $B_\varepsilon\in\operatorname{Sp}(2d,\Z)$ commuting with $A$ such that 
	$$
	\forall\mathbf{N}\geq 1,\quad M_\mathbf{N}(B_\varepsilon)M_\mathbf{N}(A)=M_\mathbf{N}(A)M_\mathbf{N}(B_\varepsilon).
	$$
	and such that, for any sequence $(\psi_k)_{k\geq 1}$ satisfying
	$$
	\forall k\geq 1,\quad M_{\mathbf{N}_k}(A)\psi_k=e^{i\beta_k(A)}\psi_k,\quad M_{\mathbf{N}_k}(B_\varepsilon)\psi_k=e^{i\beta_k(\varepsilon)}\psi_k,\quad\|\psi_k\|_{\mathcal{H}_{\mathbf{N}_k}}=1,
	$$
	one has that,
	for every $\mu\in\mathcal{P}((\psi_k)_{k\geq 1})$, 
	$$\mu=\alpha\operatorname{Leb}_{\T^{2d}}+(1-\alpha)\nu,
	$$
	with 
	$$
	\alpha\geq \frac12-\varepsilon,
	$$ 
	and $h_{\operatorname{KS}}(\nu,B)=0$ for any $B\in\langle A,B_\varepsilon\rangle\leq \operatorname{Sp}(2d,\Z)$.
	
	\end{theorem}
Here, $h_{\operatorname{KS}}(\nu,B)$ is the Kolmogorov-Sinai entropy of the
measure $\nu$ with respect to $B$~\cite[Ch.~4]{Walters82}. It is a nonnegative
number which measures how much of the complexity of $B$ is captured by $\nu$.
For instance, if $\nu$ is supported on a closed orbit, then the entropy
vanishes while it is maximal for the Lebesgue measure.

The assumption of irreducibility and that no ratio of eigenvalues
is a root of unity can be captured by the Galois group of the characteristic
polynomial.
It turns out that these are generic among
symplectic matrices and we will explain how 
to construct explicitly
matrices with this property and the additional assumption $m(A)+l(A)\geq 2$ (see \S\ref{sec:galcond}).
This assumption on the eigenvalues is made to ensure the existence of an abelian subgroup $\Lambda\leq \operatorname{Sp}(2d,\Z)$ of rank $m(A)+l(A)$ containing $A$
and the matrix $B_\varepsilon$ will be picked in a convenient way inside this subgroup. In fact, we can choose $\Lambda$ in such a way
that
$$
\forall B_1,B_2\in\Lambda,\ \forall\mathbf{N}\geq 1,\quad M_\mathbf{N}(B_1)M_\mathbf{N}(B_2)=M_\mathbf{N}(B_2)M_\mathbf{N}(B_1),
$$
and we could also consider joint eigenfunctions for \emph{all}
$M_\mathbf{N}(B)$, $B\in\Lambda$. In that case, our arguments will show that
$\alpha\geq 1/2$ (see Theorem~\ref{thm:mainreducible}). In particular, if
$m(A)+l(A)=2$, we can pick $\varepsilon=0$ in Theorem~\ref{thm:mainintro}. For $m(A)+l(A)>
2$, if we pick an arbitrary matrix $B$ in $\Lambda\setminus\langle A\rangle$
rather than $B_\varepsilon$, we will only get a lower bound $\alpha\geq
c(A,B)>0$ on the Lebesgue component with $c(A,B)$ depending on the Lyapunov
exponents of $A$ and $B$.

When $\psi_k$ is only an eigenfunction of $M_{\mathbf{N}_k}(A)$, Kim recently proved that, under a stronger assumption on the Galois group, any limit measure $\mu\in\mathcal{P}((\psi_k)_{k\geq 1})$ has full support~\cite{KimSemiclassical}. See also~\cite{Schwartz2024, DyatlovJ} for earlier contributions of Schwartz when $d=1$ and of Dyatlov and J\'ez\'equel under more restrictive assumptions than in~\cite{KimSemiclassical} when $d\geq 1$. Theorem~\ref{thm:mainintro} shows that under one extra symmetry, there is a Lebesgue component with almost $1/2$ weight.  Kurlberg, Ostafe, Rudnick and Shparlinski also proved that, for a density $1$ of integers $(\mathbf{N}_k)_{k\geq 1}$, one must have $\alpha=1$ in Theorem~\ref{thm:mainintro} under some related assumptions on the matrix $A$~\cite{KurlbergOstafeRudnickShparlinski2024}. Here we do not make any assumption on the sequence of integers $\mathbf{N}$ at the expense of considering joint eigenfunctions and of having only $\alpha\geq 1/2-o(1)$. Earlier works of the first author also show that any element $\in\mathcal{P}((\psi_k)_{k\geq 1})$ have positive entropy without any restriction on $A$ or on the sequence of integers $(\mathbf{N}_k)_{k\geq 1}$~\cite{Riviere11}. This quantitative statement will in fact be one of the key ingredients of our proof. See also~\cite{FaureNonnenmacherDeBievre2003, BonechiDeBievre2003, FaureNonnenmacher2004, AnantharamanNonnenmacher2007, Brooks2010, Gutkin2010} for earlier related results.

In \S\ref{ss:Hecke}, we will also discuss in more details results by Kelmer for joint eigenmodes of the Hecke operators~\cite{Kelmer10}. Roughly speaking, he picked joint eigenstates for all the generators of the Hecke group of $A$ while we are using only two elements in this group. Let us just mention at this point that, under the assumption of Theorem~\ref{thm:mainintro}, he proved that such joint arithmetic eigenmodes yield $\alpha=1$ in the above statement. This property is referred as Arithmetic Quantum Unique Ergodicity for quantum maps. He also showed that, if $A$ has an invariant isotropic subspace, then such a property fails even for the Hecke eigenmodes. For instance, when $A$ is of the form $\text{Diag}(\tilde{A},(\tilde{A}^{-1})^T)$ with $\tilde{A}\in\text{GL}(d,\Z)$, then one can construct a sequence of eigenmodes whose limit measure is of the form $\text{Leb}_{\mathbb{T}^d}\otimes \delta_0^{\mathbb{T}^d}$~\cite[App.~B]{Gurevich2006}. Motivated by this example, we can state a second application of the dynamical methods used in our work.
\begin{theorem}
	\label{thm:mainintro2}
	Let $d\geq 3$ and $\tilde{A}\in  \operatorname{GL}(d,\Z)$ with irreducible characteristic polynomial. Suppose also that $m(\tilde{A})+l(\tilde{A})\geq {3}$ and that no ratio of the eigenvalues of $A=\operatorname{Diag}(\tilde{A},(\tilde{A}^{-1})^T)\in\operatorname{Sp}(2d,\Z)$ is a root of unity.

	Then, for every $\varepsilon >0$, one can find $\tilde{B}_\varepsilon\in\operatorname{GL}(d,\Z)$ commuting with $\tilde{A}$ such that, letting $B_\varepsilon=\operatorname{Diag}(\tilde{B}_\varepsilon,(\tilde{B}_\varepsilon^{-1})^T)$, one has
	$$
	\forall\mathbf{N}\geq 1,\quad M_\mathbf{N}(B_\varepsilon)M_\mathbf{N}(A)=M_\mathbf{N}(A)M_\mathbf{N}(B_\varepsilon).
	$$
	and such that, for any sequence $(\psi_k)_{k\geq 1}$ satisfying
	$$
	\forall k\geq 1,\quad M_{\mathbf{N}_k}(A)\psi_k=e^{i\beta_k(A)}\psi_k,\quad M_{\mathbf{N}_k}(B_\varepsilon)\psi_k=e^{i\beta_k(\varepsilon)}\psi_k,\quad\|\psi_k\|_{\mathcal{H}_{\mathbf{N}_k}}=1,
	$$
	one has that, 
	for every $\mu\in\mathcal{P}((\psi_k)_{k\geq 1})$, 
	$$\mu=\alpha\operatorname{Leb}_{\T^{2d}}+\alpha_1\nu_1\otimes\operatorname{Leb}_{\mathbb{T}^d}+\alpha_2\operatorname{Leb}_{\mathbb{T}^d}\otimes\nu_2+(1-\alpha-\alpha_1-\alpha_2)\nu_0,
	$$
	with 
	$$
	2\alpha+\alpha_1+\alpha_2\geq 1-\varepsilon,
	$$
	and $h_{\operatorname{KS}}(\nu_1,\tilde{B}_1)=h_{\operatorname{KS}}(\nu_2,\tilde{B}_2)=h_{\operatorname{KS}}(\nu_0,B)=0$ for any $\tilde{B}_1\in\langle \tilde{A},\tilde{B}_\varepsilon\rangle\leq\operatorname{GL}(d,\Z)$, $\tilde{B}_2\in\langle \tilde{A}^T,\tilde{B}_\varepsilon^T\rangle\leq\operatorname{GL}(d,\Z)$ and $B\in\langle A,B_\varepsilon\rangle\leq \operatorname{Sp}(2d,\Z)$.
\end{theorem}
When $A$ is a general element in $\text{Sp}(2d,\Z)$ such that no ratio of eigenvalues is a root of unity, our method still allows us to prove some similar regularity statement and we refer to
Theorem~\ref{thm:mainreducible} for a precise formulation of our general and main regularity result. Roughly speaking, it will state that any semiclassical measure can be decomposed as a sum of a zero entropy measure with weight $\leq 1/2$ and a measure whose projection along invariant isotropic and symplectic subtori is the Lebesgue measure. %
Under the present form, Theorem~\ref{thm:mainintro2} already illustrates the kind of regularity property one can expect in a simplified setting. In particular, observe that if $\alpha=0$ (i.e. the measure $\mu$ has no Lebesgue component), then the zero entropy part of the measure has weight at most $\varepsilon$. Again, if we consider joint eigenmodes for the full quantum action, we can pick $\varepsilon=0$ (see Theorem~\ref{thm:mainreducible}).

\subsection{Comparison with Hecke operators for quantum maps}\label{ss:Hecke}

Looking at joint eigenmodes associated with two commuting symplectic matrices is related to the notion of Hecke operators for quantum maps as it was introduced by Kurlberg and Rudnick in~\cite{KR00} for $d=1$. This concept was extended to the case $d>1$ by Kelmer~\cite{Kelmer10} and, using the conventions from this reference, we briefly review this construction for the sake of comparison. For simplicity, we suppose that $A$ is irreducible and we refer to~\cite{Kelmer10} for the general separable case. The ring $\mathscr{D}\coloneqq\Z[X]/(\chi_A)$, where $\chi_A$ is the characteristic polynomial of $A$, can be naturally embedded in $\text{M}(2d,\Z)$ by letting $\iota:p\mapsto p(A)$. Recall that, if $B\in\text{GL}(2d,\Z)$ commutes with $A$, then $B$ is a rational polynomial in $A$ (see e.g. Lemma~\ref{la:invsubspaces}). Given $p\in\mathscr{D}$, one can define $p^*$ in such a way that $\iota(p^*)=p(A^{-1})$. Recall from~\cite[Cor.~2.2]{Kelmer10} that $\iota(p)\in\text{Sp}(2d,\Z)$ if and only the ``norm'' $\mathcal{N}(p)\coloneqq pp^*$ is equal to $1$. In Theorem~\ref{thm:mainintro}, the matrix $B_\varepsilon$ is picked inside a fixed abelian subgroup $\Lambda$ depending only on $A$ and having rank $m(A)+l(A)\geq 2$. With the above conventions, one has $\langle A,B_\varepsilon\rangle\leq \Lambda\leq \text{Ker}\ \mathcal{N}$.

According to~\cite[\S1.1.3]{Kelmer10}, $M_\mathbf{N}(B)$ depends only $B$ modulo $4\mathbf{N}$. Hence, one can introduce the map $\iota_\mathbf{N}:\mathscr{D}/4\mathbf{N}\mathscr{D}\rightarrow \text{M}(2d,\Z/4\mathbf{N}\Z)$ and the corresponding norm $\mathcal{N}_\mathbf{N}:\mathscr{D}/4\mathbf{N}\mathscr{D}\ni p\mapsto pp^* \in\mathscr{D}/4\mathbf{N}\mathscr{D}$. The Hecke group $C_A(\mathbf{N})$ is then as defined as a certain finite index subgroup of $\iota_\mathbf{N}\left(\text{Ker}\ \mathcal{N}_\mathbf{N}\right)$ with the index being bounded independently of $\mathbf{N}$. See~\cite[\S2]{Kelmer10} for more details. Kelmer then considered joint eigenfunctions for all the $M_\mathbf{N}(B)$ with $B$ lying in the Hecke group. Thanks to~\cite[Lemma~2.7]{Kelmer10}, the number of elements in $C_A(\mathbf{N})$ is bounded from below by $c_\epsilon \mathbf{N}^{d-\epsilon}$ and from above by $C_\epsilon\mathbf{N}^{d+\epsilon}$ and the lower bound plays an instrumental role in his proof of arithmetic quantum unique ergodicity. See for instance Propositions 3.6 and 3.7 in this reference.

In Theorem~\ref{thm:mainintro}, the number of (implicitly) involved unitary matrices $M_\mathbf{N}(B)$ is given by
$
\sharp\iota_\mathbf{N}(\langle A,B_\varepsilon\rangle),
$
and one has 
$$
\iota_\mathbf{N}(\langle A,B_\varepsilon\rangle)\leq \iota_\mathbf{N}(\Lambda)\leq \iota_\mathbf{N}(\text{Ker}\ \mathcal{N})\leq \iota_\mathbf{N}(\text{Ker}\ \mathcal{N}_\mathbf{N}).
$$
Hence, there are at most $\mathcal{O}(\mathbf{N}^{d+\epsilon})$ unitary matrices in that subgroup. However, the cardinal could be a priori much smaller in general and it is not clear if we could derive a good lower bound on this cardinal (that would allow to apply the arithmetic methods from~\cite{Kelmer10}) without further assumptions. Equivalently, we are picking only one (good) element $B_\varepsilon\ \text{mod}\ 4\mathbf{N}$ in the Hecke group (on top of $A\ \text{mod}\ 4\mathbf{N}$) and there is a priori no reason that these two elements generate\footnote{Recall for instance that, for $d=1$, there exist subsequences of integers $(\mathbf{N}_k)_{k\geq 1}$ for which the number of elements generated by $A$ is of order $\log \mathbf{N}_k$~\cite{KurlbergRudnick2001}.} enough elements to apply the arithmetic averaging arguments from~\cite[\S3]{Kelmer10}.  Despite that, Theorem~\ref{thm:mainintro} shows that one can already derive some equidistribution properties with only requiring to be a joint eigenmodes for $2$ elements. Moreover, the dynamical argument we develop allows to deal with more general symplectic matrices as in Theorem~\ref{thm:mainintro2} (see also Theorem~\ref{thm:mainreducible}). Finally, we refer to~\cite{RivasSaracenoOzorio2000} for numerics on the order of symplectic matrices modulo $\mathbf{N}$ when $d=2$ and to~\cite{KurlbergOstafeRudnickShparlinski2024} for lower bounds on this order along typical sequences of integers.

\subsection{Organization of the article}

In \S\ref{s:semiclassical}, we briefly review the construction of semiclassical
measures for symplectic linear automorphisms of the torus
following~\cite{DyatlovJ} and we recall what is known on the entropy of
semiclassical measures in that setting. Then, in~\S\ref{s:centralizers}, we describe
the centralizer of matrices in $\text{Sp}(2d,\Z)$.
We gather these elements with rigidity results of
Einsiedler and Lindenstrauss in \S\ref{s:rigidity} in order to derive the main
theorem of this article, namely Theorem~\ref{thm:mainreducible}. Section \ref{sec:galcond} discusses the genericity
of the assumption made on the matrix $A$ in Theorem~\ref{thm:mainintro} and
provides concrete examples of such matrices. Finally,
Appendix~\ref{a:metaplectic} is a brief reminder on the metaplectic
representation behind the construction of $M_\mathbf{N}(A)$.

\section{Semiclassical measures for abelian actions}\label{s:semiclassical}

In this section, we briefly review the quantization procedure used to look at the quantum counterpart of a symplectic linear map acting on $\T^{2d}$. We closely follow the conventions from~\cite{DyatlovJ} to which we refer for more details. We also recall the entropic results from~\cite{Riviere11} that will be used in the subsequent sections.

\begin{remark}
Compared with~\cite{DyatlovJ}, the semiclassical parameter $\mathbf{N}$ lies here in $\N$ due to the fact that we deal with higher rank actions while~\cite{DyatlovJ} allows for the more general case where $\mathbf{N}\in \frac12\N$. Yet, as they picked the convention $\mathbf{N}\in \N$ (and not $\mathbf{N}\in \frac12\N$) there is a factor $2$ that differs from their convention in this brief exposition part.
\end{remark}

\subsection{Quantum mechanics on \texorpdfstring{$\T^{2d}$}{the torus}}
\label{sec:quantummechanics}

Let $\mathbf{N}\geq 1$ be a positive integer. One can set 
$$
\mathcal{H}_{\mathbf{N}}\coloneqq\left\{u\in\mathcal{S}'(\R^d): U_w u=u\ \text{for all}\ w=(q,p)\ \in\Z^{2d}\right\},
$$
where
$$
\forall u\in\mathcal{S}(\R^d),\quad U_{q,p}u(x)\coloneqq e^{4i\pi \mathbf{N}\left\langle x,p\right\rangle}f(x-q).
$$
One can verify that this defines a finite dimensional space of dimension $(2\mathbf{N})^d$ and this space is naturally endowed with an Hilbert structure. See~\cite[Lemma~2.5]{DyatlovJ} for more details on this construction. 
\begin{remark}
We note that the general construction of these Hilbert spaces also involve a Floquet parameter $\theta\in\T^{2d}$ (which is fixed depending on $A$ and $\mathbf{N}$) and that it can be carried out for $\mathbf{N}\in\frac{1}{2}\mathbb{N}$. Here, as we aim at dealing with quantum actions, we fix $\theta=0$ and $\mathbf{N}\in\N$ so that the quantization condition~\cite[Eq.~$(2.57)$]{DyatlovJ} on $(\theta,\mathbf{N}$) required to quantize a symplectic matrix are satisfied for any $A$.
\end{remark}

Recall that a matrix $B$ in $\operatorname M(2d,\mathbb{R})$ is said to be symplectic if $B^TJB=J$ where 
$$
J\coloneqq \begin{pmatrix}
          0 &\text{Id}_d\\
          -\text{Id}_d &0
         \end{pmatrix}
,
$$
where $\text{Id}_d$ is the $d\times d$ identity matrix. The subgroup of symplectic matrices is denoted by $\text{Sp}(2d,\mathbb{R})$. To any symplectic matrix $B$ on $\R^{2d}$, one can associate its metaplectic representation $M_{\mathbf{N}}(B)$ of parameter $\mathbf{N}>0$ which acts naturally on $\mathcal{S}'(\R^d)$ and which verifies, for every $A_1$ and $A_2$ in $\text{Sp}(2d,\R)$,
\begin{equation}\label{e:representation}
 M_{\mathbf{N}}(A_1A_2)=\pm M_{\mathbf{N}}(A_1)M_{\mathbf{N}}(A_2)
\end{equation}
(see~\cite[Ch.~4]{Folland89}).
Note that this representation is projective, in the sense that it is defined up to a complex number of modulus $1$. Yet, as explained in~\cite[Th.~4.37]{Folland89}, this number can be chosen so that~\eqref{e:representation} holds true. See also Appendix~\ref{a:metaplectic} for a brief reminder.

When $\mathbf{N}$ is a positive integer and when $A\in \text{Sp}(2d,\mathbb{Z})=\text{Sp}(2d,\R)\cap \text{GL}(2d,\Z)$, one can look at the action of these matrices on the spaces $\mathcal{H}_\mathbf{N}$ defined above and verify that
$$
M_{\mathbf{N}}(A):\mathcal{H}_{\mathbf{N}}\rightarrow\mathcal{H}_{\mathbf{N}}.
$$ 
See~\cite[\S2.2.4]{DyatlovJ} for more details. Similarly, one has that 
$$
\forall a\in\mathcal{C}^\infty(\T^{2d}),\quad \Op_{(4\pi\mathbf{N})^{-1}}^w(a):\mathcal{H}_{\mathbf{N}}\rightarrow \mathcal{H}_{\mathbf{N}},
$$
where $\text{Op}_h^w(a)$ is the Weyl quantization of the symbol $a$ on $\R^{2d}$~\cite[Ch.~4]{Zworski12}. In order to emphasize the fact that one works with the restriction, we can set 
$$
\Op_{\mathbf{N}}(a)\coloneqq \Op_{(4\pi\mathbf{N})^{-1}}^w(a)|_{\mathcal{H}_{\mathbf{N}}}.
$$
A key property of this quantization procedure is the so-called Egorov property:
\begin{equation}\label{e:egorov}
M_{\mathbf{N}}(A)^{-1} \Op_{\mathbf{N}}(a)M_{\mathbf{N}}(A)=\Op_{\mathbf{N}}(a\circ A).
\end{equation}

\subsection{Semiclassical measures}

With these tools at hand and given $\psi\in\mathcal{H}_{\mathbf{N}}$ which is normalized, one can define the so-called Wigner distribution of $\psi$:
$$
W_{\psi}:a\in\mathcal{C}^\infty(\T^{2d})\mapsto \langle \Op_{\mathbf{N}}(a)\psi,\psi\rangle.
$$
From the Calder\'on-Vaillancourt Theorem on $\mathbb{R}^d$~\cite[Ch.~5]{Zworski12} combined with~\cite[Eq.~$(2.45)$]{DyatlovJ}, one has 
\begin{equation}\label{e:calderon}
\|\Op_\mathbf{N}(a)\|_{\mathcal{L}(\mathcal{H}_{\mathbf{N}})}\leq C_d\sum_{|\alpha|\leq N_d} \mathbf{N}^{-\frac{|\alpha|}{2}}\left\|\partial^\alpha a\right\|_{\mathcal{C}^0},
\end{equation}
where $C_d,N_d>0$ are constants depending only on the dimension. Hence, given a sequence
\begin{equation}\label{e:normalized}
\psi_k\in\mathcal{H}_{\mathbf{N}_k},\quad \|\psi_k\|_{\mathcal{H}_{\mathbf{N}_k}}=1,\quad\lim_{k\rightarrow\infty}\mathbf{N}_k=\infty, 
\end{equation}
the sequence $(W_{\psi_k})_{k\geq 1}$ defines a bounded sequence in $\mathcal{D}'(\T^{2d})$ and we denote by
$\mathcal{P}((\psi_k)_{k\rightarrow\infty})
$ the corresponding set of accumulation points as $k\rightarrow\infty$. Thanks to the G\aa rding inequality~\cite[Eq.~$(2.48)$]{DyatlovJ}, any accumulation point is in fact a probability measure on $\T^{2d}$. If we suppose that, in addition to~\eqref{e:normalized}, the sequence $\psi_k$ verifies
\begin{equation}\label{e:eigenmode}
\mathcal{M}_{\mathbf{N}_k}(A)\psi_k=e^{i\beta_k}\psi_k,\ \text{for some}\ \beta_k\in\R,
\end{equation}
then any measure $\mu$ in $\mathcal{P}((\psi_k)_{k\rightarrow\infty})$ verifies $A_*\mu=\mu$ thanks to the Egorov property~\eqref{e:egorov}.

The set of semiclassical measures for $A\in \text{Sp}(2d,\Z)$ is then defined as
$$
\mathcal{P}_{\text{sc}}(A):=\left\{\mu:\ \exists\ (\psi_k)_{k\geq 1}\ \text{verifying}\ \eqref{e:normalized}\ \text{and}\ \eqref{e:eigenmode}\ \text{such that}\ \mu\in\mathcal{P}((\psi_k)_{k\rightarrow\infty})\right\}.
$$
This defines a subset of the convex and compact\footnote{Recall that it is endowed with the weak-$\star$ topology induced by $\mathcal{C}^{0}(\T^{2d})$ on its dual.} set $\mathcal{P}(A)$ made of $A$-invariant probability measure on $\T^{2d}$. The following holds true
\begin{theorem}[{\cite[Thm. 1.1]{Riviere11}}]
\label{t:entropy}
 Let $A\in\operatorname{Sp}(2d,\Z)$ and set 
 $$
 \chi_+(A)\coloneqq\max\left\{\log |\lambda|:\ \lambda\in\sigma(A)\right\}.
 $$
 Then, for every $\mu\in\mathcal{P}_{\operatorname{sc}}(A)$,
 $$
 h_{\operatorname{KS}}(\mu,A)\geq\sum_{\lambda\in\sigma(A)}\max\left\{\log|\lambda|-\frac{\chi_+(A)}{2},0\right\},
 $$
 where eigenvalues are counted with multiplicity and where $h_{\operatorname{KS}}(\mu,A)$ is the Kolmogorov-Sinai entropy of the measure $\mu$ with respect to $A$.
\end{theorem}

We also refer to~\cite{AnantharamanSilberman13} for analogues of this result in the context of compact locally symmetric spaces.
See \S\ref{sec:entropy} for a reminder on the Kolmogorov-Sinai entropy.

\begin{remark} Note that the theorem in~\cite{Riviere11} is formulated for so-called quantizable matrices. This is to ensure that one can can also pick any $\mathbf{N}\in\frac12\N$ (and thus a Floquet parameter $\theta\in\T^{2d}$ adapted to $A$). As we only deal with $\mathbf{N}\in\N$, we can always pick $\theta=0$ for any choice of symplectic matrix $A$. 
\end{remark}

Finally, the set of semiclassical measures for an abelian subgroup $\Lambda\leq \operatorname{Sp}(2d,\Z)$ is defined as follows.
The subgroup $\Lambda$ is said to be \emph{higher rank quantizable} if, for every $\mathbf{N}\in \N$ and for every $A,B\in \Lambda$,
 $$
 M_\mathbf{N}(A) M_\mathbf{N}(B)=M_\mathbf{N}(B)M_\mathbf{N}(A).
 $$
\begin{remark}
	\label{r:quantizable}
	If $\Lambda\leq \operatorname{Sp}(2d,\Z)$ is abelian 
	then it is finitely generated by \cite[Cor.~2.1]{SegalPolycyclic}.
	Let $B_1,\ldots,B_k$ be generators of $\Lambda$.
	Then the finite index subgroup generated by $B_1^2,\ldots,B_k^2$ 
	is higher rank quantizable thanks to \eqref{e:representation}.
	Moreover, if $A\in \Lambda$ then
	$\langle A,B_1^2,\ldots,B_k^2\rangle$ is also a higher rank quantizable
	finite index subgroup of $\Lambda$ containing $A$.
	In particular, if $\Lambda\leq \operatorname{Sp}(2d,\Z)$ is abelian
	(and contains some element $A$)
	then there is a higher rank quantizable finite index subgroup of $\Lambda$
	(containing $A$).
\end{remark}

Suppose now that, in addition to~\eqref{e:normalized}, the sequence $(\psi_k)_{k\geq 1}$ verifies
\begin{equation}\label{e:eigenmode-action}
\forall A\in\Lambda,\quad M_{\mathbf{N}_k}(A)\psi_k=e^{i\beta_k(A)}\psi_k,\ \text{for some }\beta_k(A)\in\mathbb{R},
\end{equation}
where $\Lambda\leq \operatorname{Sp}(2d,\Z)$ is a higher rank quantizable subgroup.
We can then define \emph{the set of semiclassical measures for the abelian group $\Lambda$} as
$$
\mathcal{P}_{\text{sc}}(\Lambda)\coloneqq \left\{\mu:\ \exists\ (\psi_k)_{k\geq 1}\ \text{verifying}\ \eqref{e:normalized}\ \text{and}\ \eqref{e:eigenmode-action}\ \text{such that}\ \mu\in\mathcal{P}((\psi_k)_{k\rightarrow\infty})\right\}.
$$
From the Egorov theorem, one has $\mathcal{P}_{\text{sc}}(\Lambda)\subset\mathcal{P}_{\text{sc}}(A)\subset\mathcal{P}(A)$ for every $A\in\Lambda$. In particular, a measure $\mu\in\mathcal{P}_{\text{sc}}(\Lambda)$ is invariant under the action of $\Lambda$.

\subsection{A reminder on ergodic decomposition and Kolmogorov-Sinai entropy}
\label{sec:entropy}
Let $\mu$ be an element in $\mathcal{P}(A)$.
From the Birkoff ergodic theorem, one knows that there exists a subset $\Omega$ of $\T^{2d}$ such that $\mu(\Omega)=1$ and such that, for every $f\in\mathcal{C}^0(\T^{2d},\C)$, one can find $f^*\in L^1(\mu)$ such that
$$
\forall x\in\Omega,\quad\lim_{T\rightarrow \infty}\frac{1}{T}\sum_{k=0}^{T-1}(f\circ A^k)(x)=f^*(x).
$$
One can verify that the map 
$$
\mu_x:f\mapsto f^*(x), \quad f\in\mathcal{C}^{0}(\T^{2d}),
$$
defines an element in $\mathcal{P}(A)$ which is ergodic for $\mu$-almost every $x\in\T^{2d}$. This gives rise to the so-called \emph{ergodic decomposition} of the measure $\mu$~\cite[\S6.1]{EinsiedlerWard2011}:
\begin{equation}\label{e:ergodicdecomposition}
\mu=\int_{\T^{2d}}\mu_x\: d\mu(x). 
\end{equation}

Fix now a partition $\mathcal{B}\coloneqq \left(B_j\right)_{j=1,\ldots K}$ of $\T^{2d}$ and denote by $\mathcal{B}^{(T)}$ the refined partition made of elements of the form
$$
B_{\alpha_0}\cap A^{-1}(B_{\alpha_1})\cap\ldots\cap A^{-T+1}(B_{\alpha_{T-1}}),\quad\alpha=(\alpha_0,\alpha_1,\ldots,\alpha_{T-1})\in\{1,\ldots,K\}^T.
$$
Given $T\geq 1$, one can associate to every point $x\in\T^{2d}$ a single element $B_T(x)$ in $\mathcal{B}^{(T)}$ such that $x\in B_T(x)$. The Shannon-McMillan-Breiman theorem ensures that, for $\mu$-almost every $x\in\T^{2d}$, the limit 
$-\frac{1}{T}\ln\mu(B_T(x))$ exists and it is equal to the Kolmogorov-Sinai entropy of the measure $\mu_x$ with respect to the partition $\mathcal{B}$, i.e.
$$
h_{\operatorname{KS}}(\mu_x,A,\mathcal{B})=\lim_{T\rightarrow \infty}-\frac{1}{T}\ln\mu(B_T(x))
$$
(see~\cite[Ch.~3]{Parry69}). Moreover, the Kolmogorov-Sinai entropy of the measure $\mu$ (relative to $\mathcal{B}$) is given by
$$
h_{\operatorname{KS}}(\mu,A,\mathcal{B})=\int_{\T^{2d}}h_{\operatorname{KS}}(\mu_x,A,\mathcal{B})d\mu(x).
$$
Recall that the Kolmogorov-Sinai entropy $h_{\operatorname{KS}}(\mu,A)$ of $\mu$ is then defined as the supremum over all the finite partitions $\mathcal{B}$. For other definitions of entropy, see~\cite[Ch.~4]{Walters82}. One can show that
$$
h_{\operatorname{KS}}(\mu,A)=\int_{\T^{2d}}h_{\operatorname{KS}}(\mu_x,A)d\mu(x).
$$
Recall that these quantities are all nonnegative and that, for $\mu$-almost every $x\in\T^{2d}$,
$$
h_{\operatorname{KS}}(\mu_x,A) \leq\sum_{\lambda\in\sigma(A)}\max\left\{\log|\lambda|,0\right\}
$$
where $\sigma(A)$ is the set of eigenvalues of $A$ (counted with multiplicity)~\cite[Th.8.15]{Walters82}. In particular, as a corollary of Theorem~\ref{t:entropy}, one has 
\begin{corollary}
\label{c:entropy}
 Let $A\in\operatorname{Sp}(2d,\Z)$ and suppose that 
 $$
 \chi_+(A)=\max\left\{\log |\lambda|:\ \lambda\in\sigma(A)\right\}>0.
 $$
 Then, for every $\mu\in\mathcal{P}_{\operatorname{sc}}(A)$,
 $$
 \mu\left(\left\{x:h_{\operatorname{KS}}(\mu_x,A)>0\right\}\right)\geq\frac{\sum_{\lambda\in\sigma(A)}\max\left\{\log|\lambda|-\frac{\chi_+(A)}{2},0\right\}}{\sum_{\lambda\in\sigma(A)}\max\left\{\log|\lambda|,0\right\}},
 $$
 where eigenvalues are counted with multiplicity.
\end{corollary}

\section{Centralizers of symplectic matrices}
\label{s:centralizers}
In this section, we analyze the structure of the centralizer of a symplectic matrix with separable characteristic polynomial. The main result here is Theorem~\ref{t:structure-symplectic} describing the group structure of this centralizer. Before that, we also discuss two important cases that are used in the proof: the case of the linear group (Theorem~\ref{thm:rankGL}) and the case where the characteristic polynomial is irreducible (Theorem~\ref{thm:spforirreducible}).

\begin{definition}
	A polynomial $f\in \Q[X]$ is called \emph{separable} 
	if it has no multiple roots in its splitting field (equivalently in $\C$).
\end{definition}
\begin{remark}
	If $f\in \Q[X]$ is irreducible,
	then $f$ is separable.
	Indeed, if $\lambda$ was a root of $f$ of order $\geq 2$
	then $f(\lambda)=f'(\lambda)=0$
	which would contradict B\'ezout's identity. 

	More generally,
	if $f=p_1\cdots p_k$ with $p_i\in \Q[X]$ irreducible,
	then $f$ is separable if and only if all $p_i$ are distinct.
	Clearly, if $p^2 \mid f$ for some non-constant irreducible $p\in \Q[X]$
	then $f$ has multiple root.
	Conversely, if $f$ has a zero $\lambda$ of order $\geq 2$
	then there is $p_i$ such that $p_i(\lambda)=0$.
	Since $p_i$ is irreducible, $p_i$ has no multiple zeros 
	and hence $p_j(\lambda)=0$ for some $j\neq i$.
	By B\'ezout's identity there are $r,s\in \Q[X]$ such that $rp_i+sp_j=1$.
	This gives a contradiction when evaluating at $\lambda$.
\end{remark}

\subsection{Preliminary conventions}

We recall the following classical lemma in linear algebra:

\begin{lemma}
	\label{la:centralizerab}
Let $A\in \operatorname{GL}(m,\Q)$.
	\begin{enumerate}
		\item If the characteristic polynomial $\chi_A$ of $A$ is separable 
			then the minimal polynomial $m_A$ of $A$ coincides with its characteristic polynomial $\chi_A$.
		\item If $B\in {\operatorname{M}}(m,\Q)$ commutes with $A$
			and $m_A=\chi_A$ 
			then $B$ is a rational polynomial in $A$.
			In particular, the centralizer of $A$ in $\operatorname{GL}(m,\Q)$
			is abelian.
	\end{enumerate}
\end{lemma}
\begin{proof}
	Since $m_A\mid \chi_A\mid m_A^m$,
	$m_A$ and $\chi_A$ have the same irreducible factors.
	By the assumption of separability each factor occurs exactly once
	so that both polynomials must be equal. For item (ii), we
	use the structure theorem for finitely generated modules over principal ideal domains.
	We find that,
	for the $\Q[X]$-module $\Q^m$ (where $X\cdot v= Av$), $\Q^m\simeq \bigoplus_i \Q[X]/(p_i^{m_i})$ for pairwise distinct
	irreducible polynomials $p_i\in \Q[X]$
	with $\chi_A =m_A= \prod_i p_i^{m_i}$.
	The vector $v$ with all components equal to $1\in \Q[X]/(p_i^{m_i})$
	is cyclic,
	i.e.~$v,Av,\ldots,A^{m-1}v$ is a basis of $\Q^m$.
	Write then $Bv=\sum_{i=0}^{m-1}a_i A^iv$, $a_i\in \Q$,
	and it follows that $B = \sum_{i=0}^{m-1} a_iA^i$.
\end{proof}

	For a semigroup $G$ and $x\in G$ we denote by $G_x$ the centralizer
	$G_x\coloneqq \{g\in G\mid gx=xg\}$. In order to make notation more natural, we now work over a general finite dimensional
$\Q$-vector space $V$ instead of $\Q^m$.
The matrices with integer coefficients will be replaced by the following.
\begin{definition}
	\label{def:lattice}
	A \emph{lattice} in a finite dimensional $\Q$-vector space $V$
	is a free abelian subgroup $\Gamma$ of $(V,+)$ 
of rank $\dim V$,
i.e. $\Gamma = \Z v_1\oplus\cdots \oplus \Z v_{\dim V}$ for a basis $(v_i)$ of $V$.
We write $\operatorname{End}(\Gamma)$ for the endomorphisms $B$ of $V$
with $B(\Gamma)\subseteq \Gamma$
and $\operatorname{GL}(\Gamma)\coloneqq \{B\in \operatorname{GL}(V) \mid  B(\Gamma)=\Gamma\}$.
\end{definition}
Note that $\operatorname{GL}(\Gamma)$ are the units in the ring $\operatorname{End}(\Gamma)$.
Moreover, 
if we take a basis of $\Gamma$ then the matrix of $B\in \operatorname{GL}(\Gamma)$ and of $B^{-1}$
with respect to this basis has integer entries.
In particular, $\operatorname{GL}(\Z^{m})=\operatorname{GL}(m,\Z)$.
We also note that $\operatorname{GL}(\Gamma)=\{B\in \operatorname{End}(\Gamma) \mid 
\det B=\pm 1\}$.
Indeed, if $B\in \operatorname{GL}(\Gamma)$ then we saw that suitable matrices representing
$B$ and $B^{-1}$ have integer entries so that $\det B=\pm 1$.
Conversely, if $B(\Gamma)\subset\Gamma$ and $\det B=\pm 1$ then $q=(1-(-1)^{m}\det B \chi_B)/X\in \Z[X]$
and $B^{-1}=q(B)\in\Z[B]$.
Therefore, $B^{-1}(\Gamma) =q(B)(\Gamma)\subseteq \Gamma$ and $B\in \operatorname{GL}(\Gamma)$.

\subsection{The irreducible general linear case}\label{ss:dirichlet-unit}

 We begin with the following statement which is a consequence of Dirichlet's unit theorem.

\begin{theorem}
	\label{thm:rankGL}
	Let $A\in \operatorname{GL}(\Gamma)$ with characteristic polynomial $\chi_A$
which is irreducible over $\Q$. Suppose that $A$ has $r(A)$ real eigenvalues and $2c(A)$ eigenvalues in $\C\setminus\R$.
	Then $\operatorname{GL}(\Gamma)_A $
	is an abelian group of rank $r(A)+c(A)-1$,
	i.e.
	$\operatorname{GL}(\Gamma)_A \simeq F\times \Z^{r(A)+c(A)-1}$
	for a finite group $F$.
\end{theorem}

\begin{proof}
	We saw in Lemma~\ref{la:centralizerab} that 
	$\operatorname{End}(V)_A = \Q[A]$.
	We also have $\Q[A]\simeq \Q[X]/(\chi_A)$ via $p(A)\leftrightarrow p + (\chi_A)$ and we also verified after Definition~\ref{def:lattice} that $\Q(A)=\Q[A]$.
	This defines an algebraic number field as $\chi_A$ is irreducible.
	Moreover, $\Z[A] \subseteq \operatorname{End}(\Gamma)_A \subseteq \Q[A] \eqqcolon K$.
	Since $\Z[A]$ is a free abelian group of rank
	$\deg \chi_A = [ K\colon \Q]$,
	$\operatorname{End}(\Gamma)_A \eqqcolon \mathcal{O}$ is an order in $K$~\cite[\S I.12]{Neukirch}.
	Its units $\mathcal{O}^\times$ are those elements of $\operatorname{End}(\Gamma)_A$ whose inverse are also in $\operatorname{End}(\Gamma)_A$: this means precisely $\mathcal{O}^\times = \operatorname{GL}(\Gamma)_A$.
	By Dirichlet's unit theorem (see e.g.  \cite[Thm.~7.4, \S I]{Neukirch})
	for the maximal order $\mathcal{O}_K$,
	one knows that $\mathcal{O}_K^\times$ is isomorphic to $F' \times \Z^{r(A)+c(A)-1}$ 
	where $F'$ is a finite cyclic group consisting of roots of unity. Then
	it follows from~\cite[Th.12.12, \S I]{Neukirch} that $\mathcal{O}^\times \simeq F \times \Z^{r(A)+c(A)-1}$ 
	where $F$ is a finite group.
\end{proof}

\begin{corollary}
	\label{cor:gapGL}
	Let $\varepsilon>0$ and suppose that the assumptions of Theorem~\ref{thm:rankGL} are satisfied.
	Then there is $B\in \operatorname{GL}(\Gamma)_A$ such that
	\[
		\forall\lambda\in \sigma(B),\quad \frac{|\log |\lambda||}{\max\{|\log |\mu||\colon \mu\in \sigma(B)\}}
		\in [0,\varepsilon]\cup [1-\varepsilon,1],
	\]
with the convention $0/0=0$.
\end{corollary}
{This corollary is in fact a consequence of the construction behind Dirichlet's unit theorem and we will explain it at the end of this paragraph. Besides proving this corollary, we also aim at deriving the analogues of these results in the symplectic setting. To that aim, the arguments used to show Dirichlet's unit theorem need to be discussed (and used) to take into account the symplectic structure. Hence, we briefly recall the main lines to prove Dirichlet's unit theorem in the case where the algebraic number field $K$ is constructed from an element $A\in\text{GL}(\Gamma)$ (with $\chi_A$ irreducible). We  follow~\cite[\S I.7]{Neukirch} and refer to it for more details.

Recalling that $K=\Q(A)$ and letting $W_\R=\R^{r(A)}\times \C^{c(A)}$, we define the (multiplicative) group morphism
$$
j: K^\times\to W_\R^\times ,\quad p(A) \mapsto \left(p(\lambda_1),\ldots, p(\lambda_r), p(\mu_1),\ldots, p(\mu_c)\right), 
$$
where $(\lambda_1,\ldots,\lambda_r)$ are the real eigenvalues of $A$ and $(\mu_1,\overline{\mu}_1,\ldots,\mu_c,\overline{\mu}_c)$ are the ones in $\C\setminus\R$. Note that this defines an injective map. We also set the surjective morphism
\begin{align*}
	\ell:
	\left\{
		\begin{array}{lll}
	(W_\R^\times,\cdot)&\to &(\R^{r+c},+), \\
	(u_1,\ldots, u_r, z_1,\ldots,z_c)&
	\mapsto& \left(\log |u_1|,\ldots\log|u_r|,\log |z_1|^2,\ldots,\log |z_c|^2\right).
\end{array}
\right.
\end{align*}
For $B\in K^\times $, we denote by $N_{K/\Q}(B)\in\Q^\times$ the determinant of
the map $K\ni C\mapsto BC\in K$ viewed as a $\Q$-linear map. This is the
field norm on $K/\Q$ and one has $N_{K/\Q}=N\circ j$ where
$$
N(u_1,\ldots,u_r,z_1,\ldots,z_c)\coloneqq u_1\ldots u_r |z_1|^2 \ldots |z_c|^2.
$$
Recall now that Dirichlet's unit theorem is about the group structure of the multiplicative subgroup $\mathcal{O}^\times=\{B\in\mathcal{O}:N_{K/\Q}(B)\in\{\pm 1\}\}.$ To study this question, one sets 
$$
\Lambda\coloneqq( \ell\circ j)(\mathcal{O}^\times)\subseteq
H\coloneqq \left\{X\in\R^{r+c}:\ \sum_{j=1}^{r+c} X_j=0\right\}\simeq\R^{r+c-1},
$$
and one proves that this is a (full rank) lattice in that vector space~\cite[Th.~7.3]{Neukirch}\footnote{We note that in this reference they work
	with the maximal order $\mathcal{O}_K$
but everything works equally well with any other order $\mathcal{O}$.}.
Letting $(v_1,\ldots, v_{r+c-1})$ be a $\Z$-basis for this lattice, its preimage $(g_1,\ldots,g_{r+c-1})\in\mathcal{O}^\times$ by the map $\ell\circ j$ allows to define an abelian subgroup $G_0=\langle g_1,\ldots,g_{r+c-1}\rangle$ which induces a surjective morphism onto $(\Lambda,+)$. Then, letting $F'$ be the roots of unity lying in $\mathcal{O}$, one can prove that 
$$
1\rightarrow F'\hookrightarrow \mathcal{O}^\times \stackrel{\ell\circ j}\longrightarrow\Lambda\rightarrow 0
$$
is an exact sequence~\cite[Lemma~7.2]{Neukirch} from which we deduce that 
$$
\mathcal{O}^\times=\left\{fg^{n_1}\ldots g_{r+c-1}^{n_{r+c-1}}:\ f\in F',\ (n_1,\ldots,n_{r+c-1})\in\Z^{r+c-1}\right\}\simeq F'\times\Z^{r+c-1}.
$$
With these conventions at hand, we are ready to verify Corollary~\ref{cor:gapGL}.
\begin{proof}[Proof of Corollary~\ref{cor:gapGL}]
 For $B=p(A)\in\text{GL}(\Gamma)_A$, the eigenvalues of $B$ are exactly	given by
 $$(p(\lambda_1),\ldots,p(\lambda_r), p(\mu_1),p(\overline{\mu}_1),\ldots,p(\mu_c),p(\overline{\mu}_c)).$$
Hence, for every $\lambda\in\sigma(B)$, $\log|\lambda|$ (or $2\log|\lambda|$)
is a coordinate of $\ell\circ j(B)$. Recall now that
$\text{GL}(\Gamma)_A=\text{End}(\Gamma)_A^\times$ with $\text{End}(\Gamma)_A$
being an order in $K$. Hence, there exists a compact set $C$ in $H$ so that
$H=C+(\ell\circ j)(\text{GL}(\Gamma)_A)$. For $M\in\N$, we now set
$X_M=(M,-M,0,\ldots 0)$ to be an element	in $H$. Thanks to the above
decomposition, $X_M$ can be written as $x_M+z_M$ with $x_M\in C$ and $z_M\in
\ell\circ j(\text{GL}(\Gamma)_A)$. As $C$ is a compact subset of $H$, it is contained
inside $[-R_0,R_0]^{r+c}$ for some large enough $R_0>0$ (depending only on the lattice
$\ell\circ j(\text{GL}(\Gamma)_A)$). In particular, 
$$
z_M\in[M-R_0,M+R_0]\times[-M-R_0,-M+R_0]\times [-R_0,R_0]^{r+c-2}.
$$
By construction, there exists $B_M\in\text{GL}(\Gamma)_A$ such that $\ell\circ j(B_M)=z_M$ and $B_M$ has the expected property if we pick $M$ large enough (depending on $\varepsilon$).	
\end{proof}

}

\subsection{The irreducible symplectic case}
\label{sec:irredsymp}
In this paragraph, we let $(V,\omega)$ be a finite dimensional 
symplectic vector space over $\Q$ and 
we let $A\in \operatorname{Sp}(\Gamma)\coloneqq \operatorname{GL}(\Gamma)\cap \operatorname{Sp}(V,\omega)$ whose characteristic polynomial will be irreducible in the present \S\ref{sec:irredsymp}. We aim at proving the following symplectic analogue of the result in the previous paragraph.

\begin{theorem}
	\label{thm:spforirreducible}
	Let $A\in \operatorname{Sp}(\Gamma)$ with irreducible characteristic polynomial. Suppose that $A$ have $2m(A)$ real eigenvalues
	and $4l(A)$ eigenvalues in $\C\setminus (\R\cup \mathbb{S}^1)$. 
	Then $\operatorname{Sp}(\Gamma)_A$ is 
	an abelian group of rank $m(A)+l(A)$,
	i.e. $\operatorname{Sp}(\Gamma)_A\simeq F \times \Z^{m(A)+l(A)}$ where $F$
	is a finite group.
\end{theorem}
\begin{remark}
	We note that, if $\chi_A\in \Z[X]$ is the characteristic polynomial of $A\in \operatorname{Sp}(\Gamma)$,
	then $\chi_A$ is \emph{palindromic} of degree $2d=\dim V$
	(i.e.~$\chi_A(X^{-1})X^{2d} = \chi_A(X)$).
	Then 
	$\chi_A' = 2dX^{2d-1} \chi_A(X^{-1}) -X^{2d-2}\chi_A'(X^{-1})$.
	In particular, $\chi_A'(\pm 1) = \pm 2d\chi_A(\pm 1) - \chi_A'(\pm 1)$.
	This implies that if $\chi_A(\pm 1)=0$ then $\chi_A'(\pm 1)=0$.
	Hence, $\pm 1$ is not an eigenvalue of $A$
	if $\chi_A$ is separable.
\end{remark}
{
Clearly, $\operatorname{Sp}(\Gamma)_A \leq \operatorname{GL}(\Gamma)_A$
and hence $\operatorname{Sp}(\Gamma)_A \simeq F' \times \Z ^s$ with $F'= F\cap \operatorname{Sp}(\Gamma)$ 
and $s\leq r(A)+c(A) -1$.
It remains to find $s$.
In order to prove this theorem, we will let $\Q[A]\eqqcolon K$.

Recall that, thanks to our irreducibility assumption, $\pm 1$ do not belong to $\sigma(A)$.
Let us order the eigenvalues of $A$ as follows:
$\lambda_1,\lambda_1^{-1},\ldots\lambda_m,\lambda_m^{-1}\in \R$, $m=m(A)$, 
$\theta_1,\theta_1^{-1}=\overline{\theta_1},\ldots,\theta_k,\theta_k^{-1}\in \mathbb S^1$, $k= k(A)$ and
$\mu_1,\overline{\mu_1},\mu_1^{-1}, \overline{\mu_1}^{-1},\ldots,
\mu_l,\overline{\mu_l},\mu_l^{-1}, \overline{\mu_l}^{-1}\in \C\setminus (\R\cup
\mathbb{S}^1)$, $l=l(A)$. {With the conventions of
	\S\ref{ss:dirichlet-unit}, one has $k+2l=c$ and $2m=r$
	as well as $W_\R=\R^{2m}\times \C^{k}\times \C^{2l}$
	and we define $j$ accordingly.
	From \S\ref{ss:dirichlet-unit}, $\ell\circ j(\text{GL}(\Gamma)_A)$ is
	lattice of rank $r+c-1$ in $H\simeq\R^{r+c-1}$.
	We now define
	$
	W'_\R\coloneqq \R^m\times \R^k\times \C^{l}
	$
	and
\begin{equation*}
	\mathcal{J}\colon \left\{
		\begin{array}l
			W_\R \to W'_\R\\
		(u_1,\ldots,u_{2m},v_1,\ldots, v_{k},w_1\ldots,w_{2l})
	\mapsto \begin{array}l
	(u_1u_2,\ldots, u_{2m-1}u_{2m},\\
	\phantom{(}|v_1|^2,\ldots,|v_k|^2,\\ 
\phantom{(}w_1w_2,\ldots,w_{2l-1}w_{2l}).
	\end{array}
		\end{array}
		\right.
\end{equation*}
One has the following characterization of symplectic matrices in $K$.}
\begin{lemma}
	\[
		\operatorname{Sp}(\Gamma)_A = 
		\{B\in \operatorname{GL}(\Gamma)_A \mid \mathcal{J}(j(B))=(1,\ldots,1)\}
	\]
\end{lemma}
\begin{proof}
	For $B\in \operatorname{End}(V)$, denote by $B^\ast$ the unique endomorphism of $V$ such that
	$\omega(Bv,w)=\omega(v,B^\ast w)$ for all $v,w\in V$.

	If $B=p(A)\in \operatorname{End}(V)_A= K$
	then $B^\ast = p(A^\ast)=p(A^{-1})=p^\omega(A)$
	where\footnote{
		In \cite{Kelmer10} $p^\omega$ is denoted by $p^\ast$.
	In contrast to this notation we use $p^\ast=X^{\deg p} p(X^{-1})$.} $p^\omega = p(\frac{1-\chi}{X})\in \Q[X]$.
	Therefore, $1=B^\ast B$ if and only if $\chi_A\mid p^\omega p -1$.
	This is equivalent to the fact that $p(s^{-1})p(s)=1$ for every eigenvalue $s$ of $A$.
	These are precisely the coordinates of $\mathcal{J}\circ j(B)$.
	By observing that $p(\theta_i^{-1})=p(\overline{\theta_i})=\overline{p(\theta_i)}$
	we infer $\operatorname{Sp}(V,\omega)_A = \ker \mathcal{J}\circ j$.
	We finish the proof by intersecting with $\operatorname{GL}(\Gamma)_A$.
\end{proof}
}

\begin{proof}[Proof of Theorem~\ref{thm:spforirreducible}] Let 
$$
G\coloneqq \{w\in W_\R^\times \mid \mathcal{J}(w)=(1,\ldots,1)\}\leq (W_\R^\times ,\cdot).
$$
One has that $\ell\circ j (G)$ is a subgroup of $(\R^{2m+k+2l},+)$. By construction, one has that
\[
\ell(G)\subseteq  \{(x,-x)\mid x\in \R\}^m \times \{0\}^k \times \{(y,-y)\mid y\in \R\}^l.
\]
By putting $u_{2j+1} = e^x, u_{2j+2}=e^{-x}, v_{j}=1, w_{2j+1}=e^{y/2}, w_{2j+2}=e^{-y/2}$,
we see that equality holds. Hence, $\ell(G)\simeq \R^{m+l}$ as additive groups.
Set now $U\coloneqq G\cap j(\mathcal{O}^\times)=j(\text{Sp}(\Gamma)_A)$,
$\tilde{G}\coloneqq N^{-1}(\{\pm 1\})\leq W_\R^\times$ and $\tilde{U}\coloneqq
\tilde{G}\cap j(\mathcal{O}^\times)$. By construction, $G$ is a
subgroup of $\tilde{G}$ and $U=\tilde U \cap G$. Hence, $G/U$ embeds
into $\tilde{G}/\tilde{U}$.

It follows from the discussion on Dirichlet's unit theorem in \S\ref{ss:dirichlet-unit} that $\ell(\tilde{G})/\ell(\tilde{U})$ is compact, that $\ell(\tilde{U})$ is discrete and that $\text{ker}\,\ell|_{\tilde{U}}$ is finite. Hence, $\ell(G)/\ell(U)$ is compact and, as $\ell(U)$ is discrete, we find that $\ell(U)$ is a full rank lattice in $\ell(G)\simeq\R^{m+l}$. Therefore, $\text{Sp}(\Gamma)_A\simeq U\simeq \text{ker}\,\ell|_{U}\times\Z^{m+l}$.
\end{proof}

As for $\text{GL}(\Gamma)_A$, one has the following property as a consequence of the above construction.
\begin{corollary}
	[of proof of Theorem~\ref{thm:spforirreducible}]
	\label{cor:largegapirred}
	Let $\varepsilon>0$.
	Then there is $B\in \operatorname{Sp}(\Gamma)_A$ 
	such that 
	\[
		\frac{\max\{\log|\lambda|,0\}}{\max\{\log|\mu|:\mu\in\sigma(B)\}} \in [0,\varepsilon] \cup [1-\varepsilon,1]
	\]
	for all $\lambda\in\sigma(B)$ (with the convention $0/0=0$).
\end{corollary}

\begin{proof}
For this we observe that if $B=p(A)\in \operatorname{Sp}(\Gamma)_A$,
then the eigenvalues of $B$ are
\[
	p(\lambda_i),p(\lambda_i^{-1}),p(\theta_i),\overline{p(\theta_i)},p(\mu_i),\overline{p(\mu_i)}, p(\mu_i^{-1}),\overline{p(\mu_i^{-1}}).
\]
Therefore, one has again that $\log |\lambda|$ or $2\log|\lambda|$ (with $\lambda\in\sigma(B)$) correspond to the coordinates of $\ell(j(B))$.
We saw above that $\ell(U)$ is a lattice of full rank in $\ell(G)$.
Hence there is a compact set $C$ in $\ell(G)$ such that $C+\ell(U) = \ell(G)$.
We infer that for each $M\in \N$ there is $j(B_M)=u_M\in U$ and $x_M\in C$ such that 
$\ell(j(B_M)) + x_M = (M,-M,0,\ldots,0)\in \ell(G)$ in the case $m\geq 1$.
We have
\[
	\ell(j(B_M)) \in [M-R,M+R]\times [-M-R,-M+R]\times [-R,R]^{2m+k+2l-2}.
\]
The claim follows by picking $M$ large enough. The case $m=0$ works similarly by picking the last coordinates. If $m=l=0$ then all eigenvalues are of modulus $1$ so we have $0/0$.
\end{proof}

\subsection{The general symplectic case}\label{ss:structure-symplectic}	
We will now discuss the case of a general symplectic matrix in $\text{Sp}(2d,\mathbb{Z})$ with separable characteristic polynomial. To that aim, we first collect a few statements taken from~\cite[\S 2.2]{Kelmer10}. Given a polynomial $p\in\Z[X]$, we set $p^*(X)=X^{\text{deg} p}p(X^{-1})\in\Z[X]$. We can write 
$$
\chi_A= \prod_{i=1}^r p_i \prod_{j=1}^s \rho_j\rho_j^\ast,
$$ 
with $p_i=p_i^\ast,
\rho_j\in \Z[X]$ irreducible, pairwise distinct and $\rho_j\neq \rho_j^\ast$. Then, according to~\cite[Prop.~2.4 and Rk.~2.4]{Kelmer10}, one has 
\begin{equation}\label{e:decomposition-irreducible-symplectic}
	\Q^{2d}=\bigoplus_{i=1}^r\text{ker}\,p_i(A)\oplus \bigoplus_{j=1}^s\left(\text{ker}\,\rho_j(A)\oplus \text{ker}\,\rho_j^*(A)\right),
\end{equation}
where 
\begin{itemize}
 \item for every $i,j$, $\text{ker}\,p_i(A)$ and $\text{ker}\,\rho_j(A)\oplus\text{ker}\,\rho_j^*(A)$ are orthogonal with respect to the symplectic form,
 \item for $i\neq i'$, $\text{ker}\,p_i(A)$ and $\text{ker}\,p_{i'}(A)$ are orthogonal with respect to the symplectic form,
 \item for $j\neq j'$, $\text{ker}\,\rho_j(A)\oplus \text{ker}\,\rho_j^*(A)$ and $\text{ker}\,\rho_{j'}(A)\oplus \text{ker}\,\rho_{j'}^*(A)$ are orthogonal with respect to the symplectic form,
 \item for every $1\leq j\leq s$, $\text{ker}\,\rho_j(A)$ and $ \text{ker}\,\rho_j^*(A)$ are isotropic spaces,
 \item for every $i,j$, $\text{ker}\,p_i(A)$ and $\text{ker}\,\rho_j(A)\oplus\text{ker}\,\rho_j^*(A)$ are symplectic subspaces,
 \item for every $i,j$, $\text{ker}\,p_i(A)$, $\text{ker}\,\rho_j(A)$ and $\text{ker}\,\rho_j^*(A)$ are irreducible subspaces for the action of $A$.
\end{itemize}
\begin{remark}If we denote by $V_0^{\perp_\omega}$ the symplectic orthogonal of a linear subspace $V_0$, recall that $V_0$ is Lagrangian when $V_0^{\perp_\omega}=V_0$. When $V_0\subseteq V_0^{\perp_\omega}$ (resp. $V_0^{\perp_\omega}\subseteq V_0$), we say that $V_0$ is isotropic (resp. coisotropic). When $V_0^{\perp_\omega}\cap V_0=\{0\}$, the subspace is symplectic for $\omega|_{V_0}$.
 \end{remark}
In the following, we shall use the following convention, for every $1\leq i\leq r$ and for every $1\leq j\leq s$,
\begin{equation}\label{e:list-invariant-subspaces}
V_i\coloneqq \text{ker}\,p_i(A),\ W_j\coloneqq \text{ker}\,\rho_j(A)\oplus \text{ker}\,\rho_j^*(A),\ \overline{W}_j\coloneqq \text{ker}\,\rho_j(A),\ \text{and}\ \overline{W}_j^*\coloneqq  \text{ker}\,\rho_j^*(A).
\end{equation}

By taking sums of the subspaces appearing
in~\eqref{e:list-invariant-subspaces}, we have a description of all the
$A$-invariant subspaces of $\Q^{2d}$.
We introduce
the following sublattices:
$$
\Delta_i\coloneqq V_i\cap\Z^{2d},\ \Gamma_j\coloneqq W_j\cap\Z^{2d},\
\overline{\Gamma}_j\coloneqq \overline{W}_j\cap\Z^{2d}
\ \text{and}\ \overline{\Gamma}_j^*\coloneqq \overline{W}_j^*\cap\Z^{2d}.
$$
Their ranks are given by the dimension of $V_i$, $W_j$, $\overline{W}_j$ and $\overline{W}_j^*$ respectively.
Indeed,
if we start with a basis $v_1,\ldots,v_n$ of one of the subspaces $V_i$, $W_j$, $\overline{W}_j$ or $\overline{W}_j^*$
then there is an integer $N$ such that, for any $i$, $Nv_i$ is in the sublattice. 
This shows that $\Delta_i$, $\Gamma_j$, $\overline{\Gamma}_j$,
and $\overline{\Gamma}_j^\ast$ are lattices in the respective subspaces.
In particular, $\Gamma\coloneqq \bigoplus_i\Delta_i\oplus\bigoplus_j\Gamma_j$ has finite index inside $\mathbb{Z}^{2d}$. One has 
\begin{lemma}
	\label{la:finiteindex-lineargroup}
	Let $G$ be a subgroup of $\operatorname{GL}(2d,\R)$.
	Then $G\cap \operatorname{GL}(2d,\Z)$ has finite index in 
	$G\cap \operatorname{GL}(\Gamma)$. 
\end{lemma}
\begin{proof} Since the property of having finite index is stable under intersection with $G$,
	we can without loss of generality assume that $G=\operatorname{GL}(2d,\R)$.
	There is $N\in \N$ such that $N\Z^{2d} \subseteq\Gamma$.
	Let $g\in \operatorname{GL}(\Gamma)$ then $g^{\pm1} (\Gamma)\subseteq \Gamma$.
	Therefore, $g$ acts on the finite space $\Gamma/N\Gamma$
	so that there is $M\in \N$ such that $g^{\pm M}$ is the identity on $\Gamma/N\Gamma$.
	Hence, $g^{\pm M} = I + NX_\pm$ with $X_\pm(\Gamma)\subseteq\Gamma$.
	It follows that
	\[
		(g^{\pm M} - I)(\Z^{2d}) \subseteq X_\pm(\Gamma)\subseteq \Gamma\subseteq\Z^{2d}.
	\]
	This implies $g^M\in \operatorname{GL}(\Z^{2d})$.
	Since $\operatorname{GL}(\Gamma)$ is finitely generated
	(as it is isomorphic to $\operatorname{GL}(2d,\Z)$ by choosing a basis),
	the lemma follows.
\end{proof}
In particular, this lemma shows that $\text{Sp}(2d,\Z)_A$ has finite index in 
\begin{equation}\label{e:isomorphism}
\text{Sp}(\Gamma)_A\simeq \prod_{i=1}^r\text{Sp}(\Delta_i)_{A|_{V_i}}\times \prod_{j=1}^s\text{Sp}(\Gamma_j)_{A|_{W_j}}.
\end{equation}
The group structure of $\text{Sp}(\Delta_i)_{A|_{V_i}}$ was already described in Theorem~\ref{thm:spforirreducible}
as the characteristic polynomial of $A|_{V_i}$ is $p_i=p_i^\ast$ which is irreducible.
Hence, it remains to describe the group structure of $\text{Sp}(\Gamma_j)_{A|_{W_j}}$. To that aim, observe first that the same argument as in the proof of Lemma~\ref{la:finiteindex-lineargroup} shows that $\text{Sp}(\Gamma_j)_{A|_{W_j}}$ has finite index in $\text{Sp}(\overline{\Gamma}_j\oplus\overline{\Gamma}_j^*)_{A|_{W_j}}$. Then it remains to use the following two lemmas:

\begin{lemma}
	\label{la:structureisotropic}
	Let $\tilde{A}\in \operatorname{Sp}(W,\omega)$ with characteristic
polynomial $\chi_{\tilde{A}} = \rho \rho^\ast$ where $\rho\in \Q[X]$ irreducible and $\rho\neq \rho^\ast$. Then, there is an $\tilde{A}$-invariant Lagrangian subspace $\overline{W}\leq W$ such that
	$W\simeq \overline{W}\times \overline{W}^\ast$ (with $\overline{W}^\ast$ the dual space to $\overline{W}$) and 
	$\omega((v,\lambda),(w,\mu))=\lambda(w)-\mu(v)$.
	Moreover, $\operatorname{Sp}(W,\omega)_{\tilde{A}} \simeq \{ \diag(B,(B^{-1})^T)\mid B\in \operatorname{GL}(\overline{W})_{\tilde{A}|_{\overline{W}}}\}\simeq \operatorname{GL}(\overline{W})_{\tilde{A}|_{\overline{W}}}$.
\end{lemma}

\begin{remark}
 Under the assumption of this lemma, $\tilde{A}$ has no eigenvalues of modulus $1$.
Indeed, if $\rho(\theta)=0$ for $|\theta|=1$ then
$\rho(\overline{\theta})=0=\rho^\ast(\theta^{-1})$ which would contradict the separability as $\overline{\theta}=\theta^{-1}$.
\end{remark}

\begin{proof}
	We write $1=r\rho+s\rho^\ast$.
	Let $v,v'\in \overline W\eqqcolon \ker p(\tilde{A})$.
	Then $v=s\rho^\ast(\tilde{A})v$ and $\omega(v,v')=\omega(sp^\ast(\tilde{A}),v')=\omega(v,s^\ast p(\tilde{A})v)=0$.
	Hence $\overline W\subseteq \overline{W}^{\perp_\omega}$.
	Since $\dim \overline{W}=\dim \ker p^\ast(\tilde{A}) =\frac 12 \dim W$ 
	we have $\overline{W}=\overline{W}^{\perp_\omega}$.
	Letting $\overline{W}'=\ker \rho^\ast(\tilde{A})$ we have $\overline{W}'^{\perp_\omega} = \overline{W}'$ and 
	$W=\overline{W}\oplus \overline{W}'$.
	It follows that $w'\mapsto \omega(w',\cdot)$ defines an isomorphism $J\colon \overline{W}'\to \overline{W}^\ast$.
	For $w\in \overline{W}$ and $\lambda\in \overline{W}^\ast$,
	we have $\omega(J^{-1}(\lambda),w)=\lambda(w)$.
	Hence, $(W,\omega)$ is $\overline{W}\times \overline{W}'$ with the standard 
	symplectic form.
	Then clearly $B\in \operatorname{Sp}(W,\omega)_ {\tilde{A}}$ if and only if $B$ is of the form 
	$\diag(\overline{B},(\overline{B}^{-1})^T)$ for 
	$\overline{B} \in \operatorname{GL}(\overline W)_{\tilde{A}|_{\overline{W}}}$.
\end{proof}

\begin{lemma}
	In the situation of Lemma~\ref{la:structureisotropic}
	if $\overline \Gamma$ is a lattice in $\overline{W}$
	and $\overline \Gamma^\ast $ is a lattice in $\overline{W}^\ast$
and $\tilde A\in \operatorname{Sp}(\overline \Gamma\otimes \overline{\Gamma}^\ast)$
then $\operatorname{Sp}(\overline\Gamma\oplus \overline \Gamma^\ast)_{\tilde A}$
	is a finite index subgroup of $\operatorname{GL}(\overline \Gamma)_{\tilde A|_{\overline W}}$
	under the isomorphism $\operatorname{Sp}(W,\omega)_{\tilde A} \simeq \operatorname{GL}(\overline{W})_{\tilde A|_{\overline{W}}}$.
\end{lemma}
\begin{proof}
	As the isomorphism
	$\mathcal{R}\colon \operatorname{Sp}(W,\omega)_{\tilde A}
	\simeq \operatorname{GL}(\overline{W})_{\tilde A|_{\overline{W}}}$
	is given by restriction to $\overline{W}$,
	$\operatorname{Sp}(\overline{\Gamma}\oplus \overline{\Gamma}^\ast)_{\tilde A}$
	is mapped into $\operatorname{GL}(\overline \Gamma)_{\tilde A|_{\overline{W}}}$.
	Let us define the so-called colattice
	$\operatorname{co}(\overline{\Gamma})\coloneqq \{\lambda\in \overline{W}^\ast\mid
	\lambda(\overline{\Gamma})\subseteq\Z\}$.
	Then for $B=\diag(\overline{B},\overline{B}^{-T})\in \operatorname{Sp}(W,\omega)_{\tilde A}$
	with $\overline{B}\in \operatorname{GL}(\overline{\Gamma})$
	we have $(\overline{B}^{-1})^T(\lambda)(\gamma) = \lambda(\overline{B}^{-1} \gamma) \in \Z$
	for $\lambda\in \operatorname{co}(\overline{\Gamma})$ and $\gamma\in \overline{\Gamma}$.
	Hence, $B\in \operatorname{Sp}(\overline \Gamma \oplus \operatorname{co}(\overline{\Gamma}))$.
	It follows
	$\mathcal{R}(\operatorname{Sp}(\overline{\Gamma} \oplus \overline{\Gamma}^\ast)_{\tilde A})
	\subseteq \operatorname{GL}(\overline{\Gamma})_{\tilde A|_{\overline{W}}}
	\subseteq \mathcal{R}(\operatorname{Sp}(\overline{\Gamma} \oplus \operatorname{co}( \overline{\Gamma}))_{\tilde A})$.
	By applying again Lemma~\ref{la:finiteindex-lineargroup}
	we find that 
	$\operatorname{Sp}(\overline{\Gamma} \oplus \overline{\Gamma}^\ast)_{\tilde A}$
	has finite index in 
	$ \operatorname{Sp}(\overline{\Gamma} \oplus \operatorname{co}( \overline{\Gamma}))_{\tilde A}$.
	This implies that 
	$\mathcal{R}(\operatorname{Sp}(\overline{\Gamma} \oplus \overline{\Gamma}^\ast)_{\tilde A})
	$ has finite index in $\operatorname{GL}(\overline{\Gamma})_{\tilde A|_{\overline{W}}}$.
\end{proof}

As a consequence of this lemma and of the decomposition~\eqref{e:isomorphism}, we can deduce that $\text{Sp}(2d,\mathbb{Z})_A$ has finite index in a subgroup isomorphic to
$$
\prod_{i=1}^r\text{Sp}(\Delta_i)_{A|_{V_i}}\times \prod_{j=1}^s\text{GL}(\overline{\Gamma}_j)_{A|_{\overline{W}_j}}.
$$
Combined with Theorems~\ref{thm:rankGL} and~\ref{thm:spforirreducible} and the fact that $A|_{\overline{W}_j}$ has no eigenvalue of modulus $1$, we infer the following structure theorem:
\begin{theorem}\label{t:structure-symplectic}
	Let $A\in \operatorname{Sp}(2d,\Z)$ with separable characteristic polynomial and suppose that $A$ has $2m(A)$ real eigenvalues 
	and $4l(A)$ eigenvalues in $\C\setminus (\R\cup \mathbb{S}^1)$. Then, one has
	\[
		\operatorname{Sp}(2d,\Z)_A \simeq F \times\Z^{m(A)+l(A)-I(A)}
	\]
	where $F$ is a finite abelian group (consisting of matrices $B$ with $B^{|F|}=1$)
	and $2I(A)$ is the number of irreducible isotropic invariant subspaces of $A$ in $\Q^{2d}$.
\end{theorem}

\section{Reducibility and rigidity of actions on tori}
\label{s:rigidity}

In this section, we combine the constructions from the previous sections with a rigidity result of Einsiedler and Lindenstrauss~\cite{EL03, EinsiedlerLindenstrauss2022} in view of describing the regularity of semiclassical measures. Along the way, we also review some material from rigidity of $\Z^r$-actions on tori and describe criteria on the matrix $A$ we started with where these results apply. This allows us to prove Theorems~\ref{thm:mainintro} and~\ref{thm:mainintro2} from the introduction and to state and prove our main Theorem (Theorem~\ref{thm:mainreducible}) on semiclassical measures for joint eigenmodes.

\subsection{Preliminary conventions}

Let $m\in \N$.
We consider the action of $\text{GL}(m,\Z)$ on $\T^m$.
The following notion will be used in Theorem~\ref{thm:rigidity} below.
\begin{definition}[{\cite{EL03}}]\label{def:irred}
	\begin{enumerate}
		\item 
			The action of a subgroup $\Lambda \subseteq \operatorname{GL}(m,\Z)$ 
			on $\T^m$ is called  \emph{irreducible}
			if every infinite $\Lambda$-invariant subgroup of $\T^m$ is dense.
		\item
			The action is called \emph{totally irreducible} 
			if the action restricted to every finite index subgroup
			is irreducible.
        \item The action is called \emph{virtually cyclic} if there exists $g_0\in\Lambda$ and $\Lambda'\leq\Lambda$ with finite index such that, for every $g\in\Lambda'$, one can find $k\in\Z$ such that $g=g_0^k$.
	\end{enumerate}
\end{definition}

Observe that, if $\Lambda$ is abelian and if the rank of $\Lambda$ is $\geq 2$ then there is an injective group morphism $\rho:\Z^2\rightarrow\Lambda$, and the action of $\Lambda$ is not virtually cyclic. Indeed, otherwise, there would exist $g_0\in\Lambda$ and $\Lambda'\leq\Lambda$ of finite index such that, for all $g'\in\Lambda'$, one has $g'=g_0^k$ for some $k$. Finite index would then ensure the existence of $p\in\Z$ such that $g^p\in\Lambda'$ and $h^p\in\Lambda'$ (where $g=\rho(1,0)$ and $h=\rho(0,1)$ are given by the group morphism). One would have then $g^p=g_0^{k}$ and $h^p=g_0^{l}$ for some $k,l\in\Z$. Hence, $\rho(pl,-pk)=g^{pl}h^{-pk}=\text{Id}_m$ which would contradict the injectivity of $\rho$.

\subsection{The irreducible case}

The following theorem by Einsiedler and Lindenstrauss is the key ingredient in our classification of semiclassical measures in the irreducible case.
\begin{theorem}
	[{\cite{EL03, EinsiedlerLindenstrauss2022}}]
	\label{thm:rigidity}
	Let $\Lambda\leq \operatorname{GL}(m,\Z)$ be a totally irreducible
	abelian subgroup\footnote{We
		note that $\Lambda$ is
	finitely generated by \cite[Cor.~2.1]{SegalPolycyclic}.}
 of rank $\geq 2$.
 Let $\mu$ be an ergodic measure on $\T^m$ for the action of $\Lambda$.
	Then either $\mu$ is the Lebesgue measure or
	$h_{\operatorname{KS}}(\mu,B)=0$ for all $B\in \Lambda$.
\end{theorem}
\begin{remark}
	We remark that the original version is more general
	since it allows solenoids instead of $\T^m$,
	non-faithful actions of $\Z^r$, 
	as well as non-irreducible actions (see below). In these references, the authors made the ``not virtually cyclic'' assumption. Here, this is automatically satisfied as we supposed that the abelian subgroup $\Lambda$ has rank $\geq 2$ 
	and we consider the natural action of $\operatorname{GL}(m,\Z)$ on $\T^m$ 
	which is faithful. 
\end{remark}
When applied to semiclassical measures, this theorem combined with Corollary~\ref{c:entropy} directly yields:
\begin{corollary}\label{c:EinsiedlerLindenstrauss} Let $\Lambda \leq\operatorname{Sp}(2d,\Z)$ be an abelian subgroup which is quantizable and totally irreducible with rank $\geq 2$. Then, for any $\mu\in\mathcal{P}_{\operatorname{sc}}(\Lambda)$, one has 
$$
\mu=\alpha\operatorname{Leb}_{\T^{2d}}+(1-\alpha)\nu,
$$
where $h_{\operatorname{KS}}(\nu,\gamma)=0$ for every $\gamma\in\Lambda$ and where 
$$
\alpha\geq\max_{\gamma\in\Lambda: \chi_+(\gamma)>0}\left\{\frac{\sum_{\lambda\in\sigma(A)}\max\left\{\log|\lambda|-\frac{\chi_+(\gamma)}{2},0\right\}}{\sum_{\lambda\in\sigma(\gamma)}\max\left\{\log|\lambda|,0\right\}}\right\},
$$
with eigenvalues counted with multiplicity.
\end{corollary}
\begin{proof} We write the ergodic decomposition of $\mu$ with respect to the action of $\Lambda$, i.e. $\mu=\int_{\mathcal{E}}ed\tau(e)$ where $e$ runs over the set $\mathcal{E}$ of $\Lambda$-ergodic measures. From Theorem~\ref{thm:rigidity}, $e$ is either the Lebesgue measure or has zero entropy for every $\gamma$ in $\Lambda$. So that we can decompose $\mu=\alpha\operatorname{Leb}_{\T^{2d}}+(1-\alpha)\nu$ where $\nu$ has zero entropy for every $\gamma\in \Lambda$. Applying Corollary~\ref{c:entropy} concludes the proof. 
\end{proof}

We are now left with finding conditions ensuring that a quantizable action is totally irreducible with rank $\geq 2$. 
To that aim, we state the following lemma.

\begin{lemma}
	\label{la:invsubspaces}
	\begin{enumerate}
		\item Let $A\in \operatorname{GL}(m,\Z)$. 
			The closed connected $A$-invariant subgroups of $\T^m$
			are in one to one correspondence with
			the $A$-invariant subspaces of $\Q^m$. {The correspondence is as follows: given a closed connected $A$-invariant subgroup $T$ of $\T^m$, its tangent space is $V(T)\otimes\mathbb{R}$ where $V(T)$ is the corresponding $A$-invariant subspace of $\Q^m$.}
		\item For $A\in \operatorname{GL}(m,\Z)$
			the {only $A$-invariant subspaces of $\Q^m$ are $\{0\}$ and $\Q^m$}
			if and only if
			$A$ has irreducible characteristic 
			polynomial over $\Q$.
		\item If $A\in \operatorname{GL}(m,\Z)$ 
			has separable characteristic polynomial
			and $\Lambda\leq\operatorname{GL}(m,\Z)$ is an abelian subgroup
			containing $A$,
			then $\Lambda$ is irreducible if and only if 
			$A$ has irreducible characteristic polynomial.
	\end{enumerate}
\end{lemma}

We note that in (iii) the possible $\Lambda$ are subgroups of $\operatorname{GL}(m,\Z)_A$.
More precisely, if $\Lambda\leq \operatorname{GL}(m,\Z)$ is abelian
and $A\in \Lambda$, then
$\Lambda \subseteq \operatorname{GL}(m,\Z)_A$.
Hence, (iii) can be applied to any abelian $\Lambda\leq \operatorname{GL}(m,\Z)$
that contains some matrix with separable characteristic polynomial.
\begin{remark}
The case of invariant subsets instead of subgroups of $\T^m$ is characterized 
in \cite{Berend83}.
In this case one needs two more conditions for the absence infinite invariant closed subsets. In particular, one has to make assumptions on the eigenvalues and one needs higher rank.
\end{remark}

\begin{proof}
	The arguments we give are similar to \cite[Lemma~4.3]{DyatlovJ} (see also~\cite[App.]{KimSemiclassical}).
	For (i)
	let $T\leq \T^m$ be a connected closed $A$-invariant subgroup of $\T^m$.
	Then $T$ is a Lie subgroup with Lie algebra $\mathfrak{t}\leq \R^m$ which is $A$-invariant.
	The exponential map of $\T^m$ is just the quotient map $\pi$.
	The exponential map for $T$ is its restriction and is surjective since $T$ is connected and abelian.
	Hence, $T=\pi(\mathfrak{t})=(\mathfrak{t} + \Z^m) /\Z^m\simeq \mathfrak{t}/(\mathfrak{t}\cap \Z^m)$.
	Since $T$ is compact, $\mathfrak{t}\cap \Z^m$ is a cocompact lattice in $\mathfrak{t}$.
	{As $T$ (and thus $\mathfrak{t}$) is $A$-invariant,} we obtain that $V \coloneqq \mathfrak{t}\cap \Q^m$ is an $A$-invariant subspace of $\Q^m$.

	Conversely, if $V\leq \Q^m$ is an $A$-invariant subspace
then $\pi(V \otimes \R)$ is a connected
$A$-invariant subgroup of $\T^m$.
It is also closed as $V$ is contained in $\Q^m$.
Indeed, by Lemma~\ref{la:integercomplement} below we choose a complement $W$ of $V$ in $\Q^m$
such that $\Z^m = (V\cap\Z^m ) + (W \cap \Z^m)$.
If now $\pi(x_q)\to \pi(x)$ with $(x_q)_{q\geq 1}\in V\otimes \R$ and $x\in \R^m$
then there are $(z_q)_{g\geq 1}\in \Z^m$ such that $x_q +z_q \to x$.
But $x= v+w\in (V\otimes \R) \oplus (W\otimes \R)$ and $z_q = v_q+w_q\in  (V\cap\Z^m ) + (W \cap \Z^m)$.
Therefore, $w_q\to w\in \Z^m$ and $\pi(x)=\pi(v)\in \pi ( V\otimes \R)$.
The two constructions are clearly inverse to each other.

We now turn to the proof of the next items and we regard $\Q^m$ as $\Q[X]$-module where $X\cdot v = Av$.
Then, by the structure theorem  for finitely generated modules over principal ideal domains {and the Chinese remainder theorem}, one has
$\Q^m = \bigoplus_{i=1}^{l} \bigoplus _{j=1}^{m_i} \Q[X] / (p^{\nu_{i,j}}_i)$
where $p_i$ are irreducible and pairwise distinct and $\nu_{i,1}\leq \nu_{i,2}\leq\cdots \leq \nu_{i,{m_i}}$.
One has
\[
	\chi_A= \prod_{i=1}^l p_i^{\sum_{j=1}^{m_i} \nu_{i,j}}.
\]
The $A$-invariant subspaces of $\Q^m$ are precisely $\Q[X]$-submodules. By B\'ezout's identity, one finds
$$V_i = \bigoplus_{j=1}^{m_i} \Q[X]/\left(p_i^{\nu_{i,j}}\right) = \left\{v\in \Q^m\mid p_i^{\nu_{i,m_i}}(A) v =0\right\}.$$
Moreover, using again B\'ezout's identity, each $\Q[X]$-submodule is a direct sum of submodules of $V_i$.
{For an irreducible $p$,} the submodules of $\Q[X]/(p^\nu)$ are precisely the ones generated by $1,p,p^2,\ldots,p^\nu$.
Hence in order to have no submodules at all we must have $l=1$, $m_i=1$, and $v_{i,1}=1$,
i.e.~$\Q^m = \Q[X]/(p)$ with $p$ irreducible.
This is equivalent to saying that $\chi_A$ is irreducible over $\Q$ and (ii) is proved.
For (iii) we observe that by Lemma~\ref{la:centralizerab} the
$A$-invariant subspaces and the $\Lambda$-invariant subspaces coincide.
The claim follows by (ii).
\end{proof}
\begin{lemma}
	\label{la:integercomplement}
	Let $V\leq \Q^m$ be a subspace.
	Then there exists a complement $W$ of $V$ in $\Q^m$
	such that $\Z^m= (V\cap \Z^m) + (W\cap \Z^m)$.
\end{lemma}
\begin{proof}
	The sublattice $V\cap \Z^m$ is a subgroup of the free abelian group $\Z^m$.
	Hence there is a basis $v_1,\ldots,v_m$ of $\Z^m$,
	$k\in \N$,
	and $d_1,\ldots,d_k\in \Z$ such that
	$d_1v_1,\ldots,d_kv_k$ is a basis of $V\cap \Z^m$
	(see e.g. \cite[Thm.~4.11]{Artin}) and thus $v_1,\ldots,v_k$ is a basis of $V\cap\Z^m$). Indeed, for $i=1,\ldots k$, $d_iv_i\in V$ implies $v_i\in V\cap \Z^m$.
	It follows that there are $a_{i,j}\in \Z$
	such that $v_i=\sum_{j=1}^k a_{i,j}d_jv_j$.
	This implies $a_{i,i}d_i=1$ and therefore $d_i=\pm 1$.

	We also observe that $V=\langle v_1,\ldots,v_k\rangle_\Q$
	as for every $v\in V$ there is $N\in \Z$ with $Nv\in \Z^m$.
	The subspace $W \coloneqq \langle v_{k+1},\ldots,v_m\rangle_\Q$
	is then a complement of $V$ and we find that $\Z^m = (V\cap \Z^m) + (W\cap \Z^m)$.
\end{proof}

In order to deal with finite index subgroups we formulate the following lemma.

\begin{lemma}
	\label{la:ratiorootunity}
	Let $A\in \operatorname{GL}(m,\Q)$ with separable characteristic
	polynomial $\chi_A$. Then
	$\chi_{A^k}$ is separable for all $k\in \N$ if and only
	if no quotient of eigenvalues is a root of unity,
	If $\lambda/\mu$ for two eigenvalues $\lambda,\mu$ of $A$
	is a root of unity,
	then $(\lambda/\mu)^N=1$ for some $N\in \N$ with $\varphi(N)\leq m^2$
	where $\varphi$ is Euler's totient function.
\end{lemma}
Here and after, eigenvalues are considered with their multiplicity. In particular, if no quotient of eigenvalues is a root of unity, all eigenvalues of $A$ are distinct (which is exactly asking the characteristic polynomial to be separable).
\begin{proof}
	Let $\chi_A=\prod (X-\lambda_i)$ with $\lambda_i\in \C$.
	Then $\chi_{A^k} = \prod (X-\lambda_i^k)\in \Q[X]$.
	By assumption $\lambda_i\neq \lambda_j$ for $i\neq j$
	and we have to show that $\lambda_i^k\neq \lambda_j^k$ for $i\neq j$,
	i.e.~$(\lambda_i/\lambda_j)^k\neq 1$.
	This is the assumption of the lemma.

	If $\lambda,\mu$ are eigenvalues and $\lambda/\mu$ is a primitive $N$-th root of unity
	then
	\[
		\varphi(N) = [\Q[\lambda/\mu]\colon \Q] \leq [\Q[\lambda,\mu]\colon \Q]
		\leq [\Q[\lambda]\colon \Q]\cdot [\Q[\mu]\colon \Q] =m^2.
		\qedhere
	\]
\end{proof}
\begin{remark}
	\begin{itemize}
		\item 
			If $\sigma(A)\subseteq \R_+$
			then $\lambda/\mu\in \R_+$ for all eigenvalues $\lambda$ and $\mu$.
			Hence this ratio cannot be a non-trivial root of unity.
		\item $\varphi(N)\to \infty$ for $N\to \infty$.
			Therefore, we only have to check finitely many powers
			to apply the above lemma.
			More precisely,
			$\varphi(N)\geq \frac{N\log2}{\log(2N)}$
			for $N\geq 2$.
			We refer to \cite[Ch.~4.I.C]{RibenboimPrimeNumberRecords}
			and \cite[Thm.~15]{RosserSchoenfeld}
	for this and more explicit lower bounds.
\end{itemize}
\end{remark}

As a corollary of Lemma~\ref{la:invsubspaces}, we directly obtain the following statement.

\begin{corollary}
	\label{cor:charpolyimpltame}
		Let $A\in \operatorname{GL}(m,\Z)$  such that
			no ratio of eigenvalues is a root of unity.
			An abelian subgroup $\Lambda\leq \operatorname{GL}(m,\Z)$
			containing $A$ is totally irreducible
			if and only if $\chi_A$ is irreducible.
\end{corollary}
\begin{proof}
If $\Lambda$ is totally irreducible, 
then $\Lambda$ must be irreducible so that by Lemma~\ref{la:invsubspaces}
	$\chi_A$ is irreducible.

	Conversely, let $\chi_A$ be irreducible 
	and, using Lemma~\ref{la:invsubspaces}, let $V$ be subspace of $\Q^m$ 
	invariant under some finite index subgroup $\Lambda'$ of $\Lambda$.
	There is some $k\in \N$ such that $A^k\in \Lambda'$.
	By Lemma~\ref{la:ratiorootunity}
$\chi_{A^k}$ is separable and thus $\chi_{A^k}$ is the minimal polynomial of $A^k$.
	It follows from Lemma~\ref{la:centralizerab} that $A$ is a rational polynomial in $A^k$.
	Therefore, $V$ is $A$-invariant.
	By Lemma~\ref{la:invsubspaces} $V$ is $\{0\}$ or $\Q^m$.
\end{proof}
We are now ready to prove Theorem~\ref{thm:mainintro} from the introduction.
 \begin{proof}[Proof of Theorem~\ref{thm:mainintro}]
	 We assumed that $\chi_A$ is irreducible and no ratio of eigenvalues is root of unity.
	 Hence, by Corollary~\ref{cor:charpolyimpltame}
	 any abelian subgroup of $\operatorname{Sp}(2d,\Z)$ containing $A$
	 is totally irreducible.
	 We also assumed that $m(A)+l(A)\geq 2$.
	 By Theorem~\ref{thm:spforirreducible} and Corollary~\ref{cor:largegapirred}
	 there is, for any $\varepsilon>0$, $B_\varepsilon\in \operatorname{Sp}(2d,\Z)_A$ 
such that
$$
\forall \lambda\in\sigma(B_\varepsilon),\quad\frac{\max\{\log |\lambda|,0\}}{\chi_+(B_\varepsilon)} \in [0,\varepsilon]\cup [1-\varepsilon,1]
$$
and $\langle A,B_\varepsilon\rangle$ has rank $2$.
We now apply Corollary~\ref{cor:largegapirred} to 
a quantizable finite index subgroup of $\langle A,B_\varepsilon\rangle$
to obtain Theorem~\ref{thm:mainintro}.
\end{proof}

\subsection{The general case}

In case the irreducibility of the characteristic polynomial does not hold,
we saw that the abelian subgroup is not irreducible.
Yet, Einsiedler and Lindenstrauss showed that one can still describe $\Lambda$-ergodic measures in that case. To state this result in a concise form, we introduce the following definition.
\begin{definition}
Let $\Lambda\leq \text{GL}(m,\Z)$ be an abelian subgroup.
Let $T\leq \T^m$ be a  closed and connected $\Lambda$-invariant
subgroup\footnote{Recall from Lemma~\ref{la:invsubspaces} that this corresponds
to a $\Lambda$-invariant subspace of $\Q^m$.}. We say that a
$\Lambda$-invariant probability measure $\mu$ is an \emph{$(T,\Lambda)$-admissible measure}
if it is also $T$-invariant\footnote{Hence, the induced measure
on $T$ is the Haar measure.} and if the induced measure on $\T^m/T$ has zero
Kolmogorov-Sinai entropy for every $B\in\Lambda$.
\end{definition}
Note that, for $T=\{0\}$, $(T,\Lambda)$-admissible measures are just $\Lambda$-invariant measures with zero entropy,
whereas for $T=\T^m$ the only $(T,\Lambda)$-admissible measure is the Lebesgue measure on $\T^m$.
With this convention at hand, we can formulate the rigidity theorem as follows:

\begin{theorem}
	[{\cite{EL03,EinsiedlerLindenstrauss2022}}]
	\label{thm:rigiditynonirr}
	Let $\Lambda$ be an abelian subgroup of $\operatorname{GL}(m,\Z)$ (of
	rank $r\geq 2$) that has no virtually cyclic factors. Let $\mu$ be a
	$\Lambda$-ergodic measure on $\T^m$. Then, there exist $\Lambda'\leq
	\Lambda$ of finite index and $\Lambda'$-invariant closed connected subgroups $T_1,\ldots, T_M\leq \T^m$
	such that 
	$$
	\mu=\frac{1}{M}\left(\mu_1+\ldots+\mu_M\right),
	$$
	where each $\mu_j$ is a $\Lambda'$-ergodic measure which is $(T_j,\Lambda')$-admissible and where
	$$
	\forall\gamma\in\Lambda,\quad \gamma_*\mu_j=\mu_i\ \text{and}\ \gamma(T_j)=T_i\ \text{ for some}\ i.
	$$
\end{theorem}

Given $\Lambda\leq \text{GL}(m,\Z)$ and $T\leq \T^m$ a connected closed
 $\Lambda$-invariant subgroup {distinct from $\T^m$}, recall that a factor of $\Lambda$ is the induced
 action of $\Lambda$ on $\T^m/T$.

\begin{remark} The formulation in~\cite{EinsiedlerLindenstrauss2022} does not assume connectedness for the $(T_j)_{1\leq j\leq M}$. Yet, we can assume all the $T_j$ to be connected in this theorem.
	Indeed, if not all the $T_j$ were connected, we could proceed as follows. Since $\gamma(T_j)=T_i$, we must have $\gamma(T_j^\circ) \subseteq T_i^\circ$
	for all $\gamma\in\Lambda$
	where we denote by $T^\circ $ the connected component of the neutral element.
	It follows that $\gamma(T_j^\circ)=T_i^\circ$.
	Moreover, $\mu_j$ is $T_j^\circ$-invariant, and 
	$T_j^\circ$ is a $\Lambda'$-invariant closed and connected subgroup. Observe now that
	$T_j^\circ$ has finite index in $T_j$ so that
	$\T^m/T_j^\circ \to \T^m/T_j$ is a finite cover.
	Therefore, $h_{\operatorname{KS}}(\mu_{\T^m/T_j^\circ},\gamma)=h_{\operatorname{KS}}(\mu_{\T^m/T_j},\gamma)=0$ for all $\gamma\in\Lambda'$. In other words, $\mu_j$ is $(T_j^\circ,\Lambda')$-admissible.
 
\end{remark}

 When combined with Theorem~\ref{t:entropy}, this theorem provides conditions on semiclassical measures for quantizable actions with no virtually cyclic factors.
 Recall indeed from the proof of Lemma~\ref{la:invsubspaces} that, for a closed and connected $\Lambda'$-invariant subgroup $T\leq \T^m$, one can find a $\Lambda'$-invariant subspace $V(T)$ of $\Q^m$ such that the tangent space to $T$ is given by $V(T)\otimes\R$. One has also~\cite[Th.~8.15]{Walters82}
\begin{equation}\label{e:entropy-haar}
\forall \gamma\in\Lambda', \quad h_{\text{KS}}(\mathrm{m}_T,\gamma)=\sum_{\lambda\in\sigma(\gamma|_{V(T)})}\max\left\{\log|\lambda|,0\right\},
\end{equation}
where $\mathrm{m}_T$ is the Haar measure on $T$. Thanks to the facts that $h_{\text{KS}}(\mu_j,\gamma)=h_{\text{KS}}(\mathrm{m}_{T_j},\gamma)$ and that the Kolmogorov-Sinai entropy is an affine function, one finds that the entropy of the measure $\mu$ in the above Theorem is given by
$$
\forall\gamma\in\Lambda',\quad h_{\text{KS}}(\mu,\gamma)=\frac{1}{M}\sum_{j=1}^M\sum_{\lambda\in\text{Sp}(\gamma|_{V(T_j)})}\max\left\{\log|\lambda|,0\right\}.
$$
Hence, as $\Lambda'$ has finite index, one has that, for every $\gamma\in\Lambda$, one has $\gamma^{[\Lambda:\Lambda']}\in\Lambda'$. The previous equality translates into
$$
\forall\gamma\in\Lambda,\quad h_{\text{KS}}(\mu,\gamma)=\frac{1}{M}\sum_{j=1}^M\sum_{\lambda\in\sigma(\gamma^{[\Lambda:\Lambda']}|_{V(T_j)})}\max\left\{\log|\lambda|,0\right\}.
$$
Together with Theorem~\ref{t:entropy}, this provides constraints on the allowed semiclassical measures as any $\Lambda$-invariant measure can be decomposed as a convex sum of $\Lambda$-ergodic measures:
$$
\mu=\int_\mathcal{E}\mathrm{e}d\tau(\mathrm{e}),
$$
where $\mathcal{E}$ is the set of $\Lambda$-ergodic measures. The fact that entropy is affine implies that
$$
\forall \gamma\in\Lambda,\quad h_{\text{KS}}(\mu,\gamma)=\int_\mathcal{E}h_{\text{KS}}(\mathrm{e},\gamma)d\tau(\mathrm{e})=\int_\mathcal{E}\frac{1}{M_\mathrm{e}}\sum_{j=1}^{M_\mathrm{e}}\sum_{\lambda\in\sigma(\gamma^{[\Lambda:\Lambda_\mathrm{e}]}|_{V(T_j)})}\max\left\{\log|\lambda|,0\right\}d\tau(\mathrm{e}),
$$
which can be compared with the lower bound in Theorem~\ref{t:entropy}. Observe that these constraints are not so easy to exploit in general due to the fact that it involves a finite index subgroup $\Lambda_\mathrm{e}$ that depends on the ergodic component $\mathrm{e}\in\mathcal{E}$.

Motivated by the previous discussion, we can make the following definition that will provide simple settings to apply this theorem. We will give in \S\ref{ss:examples} below two simple (and nontrivial) examples with this property.
\begin{definition}
	 We call a subgroup $\Lambda\leq \operatorname{GL}(m,\Z)$ \emph{tame}
	if there are at most finitely many closed connected $\Lambda$-invariant subgroups of $\T^m$
	and each closed connected subgroup of $\T^m$ invariant by a finite index subgroup
	of $\Lambda$ is already invariant by $\Lambda$.
\end{definition}
 \begin{remark} When $A\in\text{Sp}(2d,\Z)$, recall from the decomposition~\eqref{e:decomposition-irreducible-symplectic} that, as soon as $\chi_A$ is separable, $\T^{2d}$ has finitely many closed connected $A$-invariant subgroups thanks to Lemma~\ref{la:invsubspaces}.
 \end{remark}
  Theorem~\ref{thm:rigiditynonirr} now reads as follows for tame subgroups.

\begin{corollary}
	\label{cor:reducibletame}
	Let $\Lambda\leq \operatorname{GL}(m,\Z)$ be a tame abelian subgroup with closed connected invariant subgroups $T_1,\ldots,T_\ell\leq \T^m$
	and with no virtually cyclic factors. 
	
	Then, for any $\Lambda$-ergodic measure $\mu$, one can find $1\leq l\leq \ell$ such that $\mu$ is $(T_l,\Lambda)$-admissible.
\end{corollary}

\begin{proof}
	Using Theorem~\ref{thm:rigiditynonirr} there is a finite index subgroup $\Lambda'\leq \Lambda$
	and $\Lambda'$-invariant closed connected subgroups $\tilde T_1,\ldots,\tilde T_M\leq \T^m$
	such that $\mu=\frac 1M (\mu_1+\cdots \mu_M)$,
	where, for each $j=1,\ldots, M$, $\mu_j$ is $\Lambda'$-ergodic, $(\tilde T_j,\Lambda')$-admissible
	and for any $\gamma\in \Lambda$ there is $i$ such that $\gamma_\ast \mu_j=\mu_i$ as well as $\gamma(\tilde T_j)=\tilde T_i$.
	Each $\tilde T_i$ is $\Lambda$-invariant as it is $\Lambda'$-invariant and $\Lambda$ is tame.
	Hence, each $\tilde T_i$ is some $T_l$.
	Let $\mu_{T_l} \coloneqq \frac 1{n_l} \sum_{i\colon \tilde T_i = T_l} \mu_i$
	with $n_l \coloneqq \#\{i\mid \tilde T_i =T_l\}$.
	Then $\mu_{T_l}$ is a $\Lambda$-invariant probability measure
	as $\gamma\in \Lambda$ permutes the summands of $\mu_{T_l}$.
	Moreover, $\mu_{T_l}$ is $(T_l,\Lambda')$-admissible
	as a sum of $(T_l,\Lambda')$-admissible measures.
	Since $\Lambda'$ has finite index in $\Lambda$ and 
	$T_l$ is $\Lambda$-invariant,
	$\mu_{T_l}$ is $(T_l,\Lambda)$-admissible.
	Moreover, $\mu=\sum_l \frac {n_l}{M} \mu_{T_l}$.
	Since $\mu$ is $\Lambda$-ergodic and $\mu_{T_l}$ is $\Lambda$-invariant,
	we infer that $\mu=\mu_{T_l}$ for some $l$.
\end{proof}

In the setting of Corollary~\ref{cor:reducibletame}, any $\Lambda$-invariant measure $\mu$ can be decomposed as 
$$
\mu=\sum_{l=1}^\ell\alpha_l\mu_l,
$$
where $\sum_{l=1}^\ell\alpha_l=1$ and where each $\mu_l$ is an $(T_l,\Lambda)$-admissible measure. The entropy can then be written as
\begin{equation}\label{e:entropy-tame}
h_{\text{KS}}(\mu,\gamma)=\sum_{l=1}^\ell\alpha_l\sum_{\lambda\in\sigma(\gamma|_{V(T_l)})}\max\{\log|\lambda|,0\},
\quad \gamma\in \Lambda,
\end{equation}
which can be more easily compared with the lower bound in Theorem~\ref{t:entropy}. Recall that one of the $T_l$ is reduced to $\{0\}$ in view of allowing zero entropy measures (like the one carried by the neutral element of $\T^m$).

Regarding semiclassical measures, we obtain the following analogue of Corollary~\ref{c:EinsiedlerLindenstrauss} as a direct consequence of Theorem~\ref{t:entropy} and~\eqref{e:entropy-tame}.

\begin{corollary}
	\label{cor:semiclassicalmeasuresreducible}
	Let $\Lambda \leq\operatorname{Sp}(2d,\Z)$ be an abelian subgroup 
	which is quantizable and tame and which has no virtually cyclic factors. 
	Let $T_1,\ldots,T_\ell$ be the closed connected subgroups of $\T^{2d}$.
	Then, for any $\mu\in\mathcal{P}_{\operatorname{sc}}(\Lambda)$, one has 
$$
\mu=\sum_{l=1}^\ell \alpha_l \mu_l \quad \text{with} \quad \sum_{l=1}^\ell\alpha_l =1,
$$
where each $\mu_l$ is $(T_l,\Lambda)$-admissible and where, for any $\gamma\in \Lambda$,
$$
\sum_{l=1}^\ell \alpha_l 
\left(\sum_{\lambda\in\sigma(\gamma|_{V(T_l)})}\max\left\{\log|\lambda|,0\right\}\right)
\geq
\sum_{\lambda\in\sigma(\gamma)}\max\left\{\log|\lambda|-\frac{\chi_+(\gamma)}{2},0\right\}.
$$
\end{corollary}

Observe that, if without loss of generality $T_\ell=1$ then $\alpha_1+\cdots +\alpha_{\ell-1} >0$. The following lemma ensuring tameness can be viewed as the analogue of Corollary~\ref{cor:charpolyimpltame} for the reducible case.
\begin{lemma}
	\label{cor:sepimpltame}
		Let $A\in \operatorname{GL}(m,\Z)$  such that
			no ratio of eigenvalues is a root of unity.
			Then, any abelian subgroup $\Lambda\leq \operatorname{GL}(m,\Z)$
			containing some power of $A$ is
			tame.
			More precisely, 
			the closed connected $\Lambda$-invariant subgroups correspond
			to direct sums of $\ker p_i(A)$ 
			where $p_i\in \Z[X]$ are the irreducible factors of $\chi_A$.
\end{lemma}
\begin{proof}
	As in the proof of Corollary~\ref{cor:charpolyimpltame},
	each subspace invariant under a finite index subgroup of $\Lambda$
	is already $A^N$-invariant for some power $N$.
	But, as $\chi_{A^N}$ is separable, by Lemma~\ref{la:ratiorootunity},
	such a subspace is $A$-invariant by Lemma~\ref{la:centralizerab}.
	Hence, we only have to determine the $A$-invariant subspaces in $\Q^m$.
	We use again the structure theorem for the $\Q[X]$-module $\Q^m$
	to see that $\Q^m \simeq \bigoplus_i \Q[X]/(p_i)$ where $p_i\in \Z[X]$
	are the irreducible factors of $\chi_A$.
	The invariant submodules are direct sums of the $\Q[X]/(p_i) \simeq \ker p_i(A)$.
	Hence there are only finitely many.
\end{proof}

We are now ready to state our main theorem in the reducible case.
\begin{theorem}
	\label{thm:mainreducible}
		Let $A\in \operatorname{Sp}(2d,\Z)$  such that
			no ratio of eigenvalues is a root of unity.
	Let $\chi_A = \prod_{i=1}^r p_i\prod_{j=1}^s \rho_j\rho_j^\ast$, $V_i$, $W_j$, $\overline{W}_j$
	as in~\eqref{e:list-invariant-subspaces}.
	Assume that, for every $1\leq i\leq r$ and $1\leq j\leq s$, 
	\begin{equation}\label{e:eigenvalue-hyp-reducible}
	m(A|_{V_i}) + l(A|_{V_i})\geq 2\ \text{and}\  m(A|_{W_j})+l(A|_{W_j})\geq 3.
	\end{equation}
	Then, for any $\mu\in \mathcal{P}_{\mathrm{sc}}(\Lambda)$,
	with $\Lambda\leq \operatorname{Sp}(2d,\Z)_A$ quantizable and of finite index,
	one has 
	\[
\mu=\sum_{I\subseteq \{V_i,\overline{W}_j,\overline{W}_j^\ast\}}
\alpha_I \mu_I 
\quad\text{with}
		\quad
		\sum_{I\subseteq\{V_i,\overline{W}_j,\overline{W}_j^\ast\}} \alpha_I =1,
	\]
	where each $\mu_I$ is $(T_I,\operatorname{Sp}(2d,\Z)_A)$-admissible
	where $T_I\coloneqq \bigoplus_{U\in I} U\otimes\R/(\Z^{2d}\cap \bigoplus_{U\in I} U\otimes\R)$
	and where, 
	\[
		\sum_{V_i\in I} \alpha_I \geq 1/2
		\quad \text{and}\quad
		\sum_{\overline W_j,\overline{W}_j^\ast\in I} \alpha_I
		+
		\frac 12 \left(\sum_{\overline W_j\in I,\overline{W}_j^\ast\notin I} \alpha_I
		\right)		+
		\frac 12 \left(\sum_{\overline W_j\notin I,\overline{W}_j^\ast\in I} \alpha_I
		\right)	\geq \frac 12
	\]
	for any choice of $V_i$ and $\overline{W}_j$.
\end{theorem}

Note also that, for a given $j$, the assumption $m(A|_{W_j})+l(A|_{W_j})\geq 3$ is satisfied as soon as $\text{dim}\ \overline{W}_j\geq 5$ (recall that separability implies that $\rho_j$ does not cancel on $\mathbb{S}^1$).

\begin{proof}
	Using Corollary~\ref{cor:semiclassicalmeasuresreducible} and
	Lemma~\ref{cor:sepimpltame} we only have to justify 
	assumption of having no virtually cyclic factors and 
	the resulting bounds on $\alpha$.
	For both we can without loss of generality assume that $\Lambda=\operatorname{Sp}(2d,\Z)_A$.
	
	For the prior, 
	we consider the induced action $\rho_{\T^{2d}/T}$ of $\operatorname{Sp	}(2d,\Z)_A$
	on the factor $\T^{2d}/T$. 
	Let $V\leq \Q^{2d}$ the $A$-invariant subspace corresponding to $T$
	and $W$ an $A$-invariant complement of $V$.
	Such a complement exists as $V$ is the direct sum of some $V_i,\overline{W}_j,\overline{W}_j^\ast$
	and $W$ is chosen as the sum of the remaining ones.
Let $\operatorname{cl}(W)$ denote the smallest $A$-invariant symplectic subspace of $\Q^{2d}$ containing $W$,
i.e.~$\operatorname{cl}(W)$ is obtained from $W$ by adding all 
$\overline{W}_j^\ast$ if $\overline{W}_j \subseteq W$ and 
$\overline{W}_j$ if $\overline{W}_j^\ast \subseteq W$. We then consider the action $\rho_{\T^{2d}/T}$ restricted to the subgroup $\operatorname{Sp}(\operatorname{cl}(W)\cap \Z^{2d})_{A|_{\operatorname{cl}(W)}}$ of $\operatorname{Sp}(2d,\Z)_A$. If $\rho_{\T^{2d}/T}$ restricted to this subgroup is injective, then we can combine the eigenvalue assumption~\eqref{e:eigenvalue-hyp-reducible} together with Theorem~\ref{t:structure-symplectic}. This implies that $\rho_{\T^{2d}/T}(\operatorname{Sp}(\operatorname{cl}(W)\cap \Z^{2d})_{A|_{\operatorname{cl}(W)}}$ has rank $\geq 2$ 
	and is therefore not virtually cyclic.

	In order to prove injectivity, we start by observing
	$\T^{2d}/T = \R^{2d}/(V\otimes \R +\Z^{2d})$.
	By Lemma~\ref{la:integercomplement} we can choose a complement $W'$ in $\Q^{2d}$ of $V$
	such that $\Z^{2d} = (V\cap\Z^{2d}) \oplus (W'\cap \Z^{2d})$.
	Then $V\otimes \R +\Z^{2d}=V\otimes \R +  W'\cap \Z^{2d}$ and
	$V\otimes \R \cap W'\cap \Z^{2d} = \{0\}$. Assume now that one has $B\in \operatorname{Sp}(\operatorname{cl}(W)\cap \Z^{2d})_{A|_{\operatorname{cl}(W)}}$
	with $\rho_{\T^{2d}/T}(B)=\text{Id}_{\T^{2d}/T}$.
	Then $h\colon  x\in\R^{2d}\mapsto Bx-x\in\R^{2n}$ is linear 
	and it has values in $V\otimes \R + (W'\cap \Z^{2d})$.
	Therefore, if we denote by $\operatorname{pr}$ the projection $\R^{2d} = (V\otimes \R) \oplus (W'\otimes \R) \to W'\otimes \R$,
	then $\operatorname{pr}\circ h \colon \R^{2d}\to W'\otimes\R$ 
	is linear and has values in $W'\cap \Z^{2d} $
	and we infer that it must vanish.
	Consequently, $h(\R^{2d})\subseteq V\otimes\R$ and also $h(\Q^{2d})\subseteq V$.
	This implies $B|_{W}=\text{Id}$
	and since $B$ is symplectic we also have $B|_{\operatorname{cl}(W)}=1$ 
	by Lemma~\ref{la:structureisotropic}.
	This proves injectivity and the above discussion shows that the action has no virtually cyclic factors.

	For the bounds on $\alpha$, recall that $\text{Sp}(2d,\mathbb{Z})_A$ has finite index in
\begin{equation}\label{e:product-structure}
\prod_{i=1}^r\text{Sp}(\Delta_i)_{A|_{V_i}}\times \prod_{j=1}^s\text{GL}(\overline{\Gamma}_j)_{A|_{\overline{W}_j}},
\end{equation}
with the conventions of \S\ref{ss:structure-symplectic}.
In particular, recall that here $\prod_{j=1}^s\text{Sp}(\overline{\Gamma}_j\oplus \overline{\Gamma}_j^*)_{A|_{W_j}}$ is identified with a subgroup of $\prod_{j=1}^s\text{GL}(\overline{\Gamma}_j)_{A|_{\overline{W}_j}}$  though the map $\overline{B}\mapsto (\overline{B},(\overline{B}^{-1})^T)$.

We first deal with the case of an	
invariant symplectic subspace $V_i$. Let $0<\varepsilon<1/2$. By Corollary~\ref{cor:largegapirred}, there is $B_i\in\text{Sp}(\Delta_i)_{A|_{V_i}}$, such that, for every $\lambda\in\sigma(B_i)$,
$$
\frac{\max\{\log|\lambda|,0\}}{\chi_+(B_i)}\in [0,\varepsilon]\cup [1-\varepsilon,1].
$$ 
For all the factors in \eqref{e:product-structure} different from $V_i$, we pick the matrix to be the identity. This yields a matrix $B$ on the product space and, thanks to finite index, we can find some $N\geq 1$ such that $B^N$ belongs to $\Lambda$. Applying the bound on $\alpha$ from Corollary~\ref{cor:semiclassicalmeasuresreducible} (for $\gamma=B^N$)
	turns into
	\[
		\left(\sum_{V_i\in I} \alpha_I\right) \left(\sum_{\lambda\in\sigma(B_i)} \max\{\log|\lambda|,0\}\right) 
		\geq \sum_{\lambda\in\sigma(B_i)}\max\left\{\log|\lambda| - \frac{\chi_+(B_i)}2,0\right\}.
	\]
	It follows that $\sum_{V_i\in I} \alpha_I \geq \frac 12-C \varepsilon$, where $C$ depends only on $\text{dim}\ V_i$. As this is valid for any $\varepsilon>0$, we get the expected lower bound.

For the other bound, consider $\overline{W}_j\leq W_j$.
Applying Corollary~\ref{cor:gapGL} instead of Corollary~\ref{cor:largegapirred}, we find some $B\in \Lambda$ acting trivially on all other $V_i$
and $W_{j'}$ and
such that 
\begin{equation}\label{e:size-eigenvalue}\forall\lambda\in\sigma(B|_{\overline{W}_j}),\quad\frac{|\log |\lambda||}{\max\{\chi_+(B|_{\overline{W}_j}),\chi_+(B|_{\overline{W}_j^*})\}}
\in [0,\varepsilon]\cup [1-\varepsilon,1].
\end{equation}
We also observe that $\sigma(B|_{\overline{W}_j^*})=\{\lambda^{-1}:\lambda\in \sigma(B|_{\overline{W}_j})\}$ and thus
$$
\frac12\sum_{\lambda\in\sigma(B|_{W_j})}\max\{\log|\lambda|,0\}=\sum_{\lambda\in\sigma(B|_{\overline{W}_j})}\max\{\log|\lambda|,0\}=\sum_{\lambda\in\sigma(B|_{\overline{W}_j^*})}\max\{\log|\lambda|,0\}.
$$
We also remark that $\chi_+(B|_{W_j})=\max\{\chi_+(B|_{\overline{W}_j}),\chi_+(B|_{\overline{W}_j^*})\}$. The lower bound from Corollary~\ref{cor:semiclassicalmeasuresreducible} yields 
\begin{multline*}
\left(\sum_{\overline{W}_j,\overline{W}_j^\ast\in I} \alpha_I+\frac12\sum_{\overline{W}_j\in I,\overline{W}_j^\ast\notin I} \alpha_I+\frac12\sum_{\overline{W}_j\notin I,\overline{W}_j^\ast\in I} \alpha_I\right)
		\sum_{\lambda\in\sigma(B|_{W_j})}\max\{\log|\lambda|,0\}\\
		\geq \sum_{\lambda\in\sigma(B|_{W_j})}\max\left\{\log|\lambda|-\frac{\chi_+(B|_{W_j})}{2},0\right\}.
\end{multline*}
Combined with~\eqref{e:size-eigenvalue}, this yields a lower bound of size $\frac12-C\varepsilon$ where $C>0$ depends only on the dimension of $W_j$. As this valid for any $\varepsilon>0$, we obtain the expected lower bound.
\end{proof}

\begin{remark}
	As in the irreducible case (Theorem~\ref{thm:mainintro}),
	we could have picked in each symplectic factor one matrix $B_\varepsilon$
	and considered the subgroup generated by them to obtain a version similar to 
	Theorem~\ref{thm:mainintro}. See for instance Theorem~\ref{thm:mainintro2} for such a formulation in the case where $r=0$ and $s=1$.
\end{remark}

\subsection{Examples in the reducible case}\label{ss:examples}
Let us give two examples to illustrate the use of Theorem~\ref{thm:mainreducible}.

\subsubsection{Example with Lagrangian invariant subspaces}
\label{ex:Lagranian}
Let $A\in \operatorname{Sp}(2d,\Z)$ with characteristic polynomial $\chi_A=pp^\ast$,
$p\in \Z[X]$ irreducible and $p\neq p^\ast$.
By Lemma~\ref{la:structureisotropic} and without loss of generality, $A=\diag(A',A'^{-T})$
with $\chi_{A'}=p$.
The assumption $p\neq p^\ast$ means that $A'$ is not itself symplectic
with respect to some symplectic form on $\Q^d$.
Let us also assume that $\ker p(A) \cap \Z^{2d} + \ker p^\ast(A)\cap \Z^{2d}=\Z^{2d}$
so that $A'\in \operatorname{GL}(d,\Z)$.
In general this holds up to finite index. The assumption that no ratio of eigenvalues of $A$ is a root of unity
transforms to no ratio and no product of eigenvalues of $A'$
is a root of unity.
The invariant subtori are $0, \T^n\times 0,0\times \T^n,\T^{2n}$.
If we assume $m(A)+l(A)\geq 3$,
i.e. $\#(\operatorname{Sp}(A')\cap \R) + \frac 12 \#(\operatorname{Sp}(A')\cap(\C\setminus \R)) \geq 3$,
we obtain that for every $\mu\in \mathcal{P}_{\mathrm{sc}}(\Lambda)$,
where $\Lambda\leq\operatorname{Sp}(2d,\Z)_A$ is quantizable and has finite index,
we have
\[
	\mu=\alpha \operatorname{Leb}_{\T^{2d}} +\alpha_2 \operatorname{Leb}_{\T^d}\otimes \nu_2 + \alpha_1 \nu_1\otimes \operatorname{Leb}_{\T^d} +\alpha_0\nu_0
\]
with $\alpha+\alpha_2+\alpha_1+\alpha_0=1$ 
and, for any $B'\in \operatorname{GL}(d,\Z)_{A'}$,
$h_{\operatorname{KS}}(\nu_2, B^T)=h_{\operatorname{KS}}(\nu_1,B)=h_{\operatorname{KS}}(\nu_0,\diag(B',B'^{-T})) = 0$.
The bound on $\alpha$ rephrases to 
\begin{equation}\label{e:entropy-symplectic-lift}
\alpha+\frac 12\alpha_1+\frac 12\alpha_2
	\geq \frac 12.
\end{equation}

\subsubsection{Example with symplectic invariant subspaces}

	In this second example, we consider the product situation.
	We consider the symplectic form $\omega$ on $\R^{2d_1+2d_2}$ given by the symplectic product structure. 
	For $i=1,2$, we let $A_i\in \operatorname{Sp}(2d_i,\Z)$ with irreducible distinct
	characteristic polynomial.
	Then $A\coloneqq \diag (A_1,A_2)\in \operatorname{Sp}(2d_1+2d_2,\Z) $ has separable characteristic polynomial
	$\chi_{A_1}\chi_{A_2}$.
	The 
	closed connected $A$-invariant subgroups are
	$0,0\times \T^{2d_2},\T^{2d_1}\times 0, \T^{2d_1+2d_2}$. 
	Again we assume that no ratio of eigenvalues of $A$ is a root of unity.
	In addition to that we assume 
	$m(A_i)+l(A_i)\geq 2$ for both $i=1,2$.
	Then 
we obtain that, for every $\mu\in \mathcal{P}_{\mathrm{sc}}(\Lambda)$,
where $\Lambda\leq\operatorname{Sp}(2d_1+2d_2,\Z)_A$ is quantizable and has finite index,
we have
\[
	\mu=\alpha \operatorname{Leb}_{\T^{2d_1+2d_2}} +\alpha_2 \operatorname{Leb}_{\T^{d_1}}\otimes \nu_2 + \alpha_1 \nu_1\otimes \operatorname{Leb}_{\T^{d_2}} +\alpha_0\nu_0
\]
with $\alpha+\alpha_2+\alpha_1+\alpha_0=1$ 
and, for any $B_i\in \operatorname{Sp}(2d_i,\Z)_{A_i}$,
$h_{\operatorname{KS}}(\nu_2, B_2)=h_{\operatorname{KS}}(\nu_1,B_1)=h_{\operatorname{KS}}(\nu_0,\diag(B_1,B_2)) = 0$.
The bound on $\alpha$ rephrases to 
\begin{equation*}
\alpha+\alpha_i
\geq \frac 12 \quad \text{for}\quad i=1,2.
\end{equation*}

\section{The Galois condition in \texorpdfstring{$\operatorname{Sp}(2d,\Z)$}{Sp(2d,Z)} and some examples}
\label{sec:galcond}
In this section, we recall a criterion for the irreducibility of
the characteristic polynomials in $\text{Sp}(2d,\Z)$ due to Anderson and Oliver \cite[Appendix~B]{KimSemiclassical}
which will also imply that no ratio of eigenvalues is a root of unity.

\subsection{A criterion for irreducibility and separability}

If $\chi$ is the characteristic polynomial of some element in $\text{Sp}(2d,\Z)$
then $\chi$ is palindromic or reciprocal,
i.e.~the coefficients of $\chi=\sum_{i=0}^{2d} a_i X^i$ satisfy
$a_i=a_{2d-i}$ for all $i$.
As a consequence, the roots are of the form 
$\lambda_1,\ldots,\lambda_d,\lambda_1^{-1},\ldots,\lambda_d^{-1}$.
Every field automorphism $\sigma\in \text{Gal}(\chi)$ must send 
$\lambda_i^{-1}$ to $\sigma(\lambda_i)^{-1}$.
This shows that the Galois group preserves the set of unordered pairs 
$\{\lambda_1,\lambda_1^{-1}\},\ldots,\{\lambda_d,\lambda_d^{-1}\}$.
The wreath product $S_2\wr S_d$ is defined as the subgroup of $S_{2d}$
preserving this set of unordered pairs
so that $\text{Gal}(\chi)\leq S_2\wr S_d$.
Hence, the largest possible Galois group is $S_2\wr S_d$ meaning $|\text{Gal}(\chi)|\leq|S_2\wr S_d| =2^dd!$.
We say that $A\in \text{Sp}(2d,\Z)$ or $\chi$ satisfies the \emph{Galois condition} if
\begin{equation}
	\label{eq:Galcond}
	\tag{G}
	|\text{Gal}(\chi)|=2^dd!.
\end{equation}

\begin{remark}
 	If~\eqref{eq:Galcond} holds for a palindromic polynomial $\chi$ with coefficients in $\Z$, 
	then $\chi$ is irreducible over $\Q$ (thus over $\Z$ if $a_0=1$).
	Indeed, let $\lambda_1^{\pm 1}, \ldots, \lambda_d^{\pm 1}$ be the roots of $\chi$.
	Then $[\Q(\lambda_i)\colon\Q]\leq 2d$
	since $\chi$ is a polynomial over $\Q$ such that $\chi(\lambda_i)=0$.
	But then $\frac{\chi(X)}{(X-\lambda_i)(X-\lambda_i^{-1})}\in \Q(\lambda_i)[X]$
	is of degree $2(d-1)$ and annihilates $\lambda_j$, $i\neq j$.
	Therefore,
	$[\Q(\lambda_i,\lambda_j)\colon \Q(\lambda_i)]\leq 2(d-1)$.
	Inductively,
	$$
	[\Q(\lambda_{\sigma(1)},\ldots,\lambda_{\sigma(r)})\colon 
	\Q(\lambda_{\sigma(1)},\ldots,\lambda_{\sigma(r-1)})]\leq 2(d-r+1),
	$$
	for $\sigma\in S_d$.
	In particular,
	since for $Z\coloneqq\Q(\lambda_1,\ldots,\lambda_d)$,
	\[
		[Z\colon \Q]
		=[\Q(\lambda_1,\ldots,\lambda_d)\colon \Q(\lambda_1,\ldots,\lambda_{d-1})]\cdots
		[\Q(\lambda_1,\lambda_2)\colon \Q(\lambda_1)]\cdot [\Q(\lambda_1)\colon \Q]
		\leq 2^d d!.
	\]
	Hence under the assumptions of the lemma, all estimates are actually equalities.
	In particular $[\Q[\lambda_i]\colon \Q]=2d$ {for every $1\leq i\leq d$}
	so that $\chi$ is irreducible. Yet, it is worth noticing that $\chi$ may be
	irreducible without~\eqref{eq:Galcond} being satisfied. For instance,
	$\chi=X^4+X^3+X^2+X+1$ is irreducible with Galois group equal to
	$(\Z/5\Z)^\times$.
\end{remark}

The following lemma guarantees the applicability of Theorem~\ref{thm:mainintro}
if additionally $m(A)+l(A)\geq 2$ holds.
\begin{lemma}
	[{\cite[Lemma~B.2]{KimSemiclassical}}]
	\label{la:fullGaloisImplieshyper}
	Let $d\geq 2$. 
	If $A\in \operatorname{Sp}(2d,\Z)$ satisfies \eqref{eq:Galcond},
	then $\chi_A$ is irreducible and 
	no ratio of eigenvalues of $A$ is a root of unity.
\end{lemma}
\begin{proof}
	The latter statement is equivalent to $\chi_{A^k}$ being separable
	for all $k\in \N$ by Lemma~\ref{la:ratiorootunity}.
	In \cite[Lemma~B.2]{KimSemiclassical} it is shown that 
	if $A$ satisfies \eqref{eq:Galcond}
	then $\chi_{A^k}$ is not only separable but even irreducible for all $k\in \N$.
\end{proof}

\begin{remark}
	We note that \eqref{eq:Galcond} does not imply $l(A) +m(A)\geq 2$.
	\[
		\begin{pmatrix}
			0&0&1&0\\
			0&0&0&1\\
			-1&0&0&1\\
			0&-1&1&2
		\end{pmatrix}
	\]
	has characteristic polynomial $x^4-2x^3 +x^2-2x+1$
	which has two real roots and two roots on the unit circle
	and has Galois group $S_2\wr S_2$~\cite[App.~A]{DyatlovJ}.
	We also note that, inside the set of all palindromic polynomials, this is quite common.
	More precisely,
	if $f(x)=a_0+\cdots + a_{2n} x^{2n}$ with $a_k=a_{2n-k}\in \R$ and 
	$|a_k|\geq |a_n|\cos \left( \frac{\pi}{[\frac{n}{n-k}]+2}\right)$
	for some $k=0,\ldots,n-1$,
	then $f$ has a root on the unit circle~\cite{KonvalinaMatachePalindrome}.
\end{remark}

In view of Example~\ref{ex:Lagranian}
we also formulate the following lemma.
\begin{lemma}
	If $A'\in \operatorname{GL}(d,\Z)$, $d\geq 3$, satisfies $\operatorname{Gal}(\chi_{A'})=S_d$
	then $A\coloneqq \diag(A',A'^{-T})\in \operatorname{Sp}(2d,\Z)$ has 
	separable characteristic polynomial $\chi_{A'}\chi_{A'}^\ast$ 
	with both factors being irreducible
	and no ratio of eigenvalues of $A$ is a root of unity.
	If $d=2$ and in addition $\chi_{A'}$ is neither $X^2+1$, $X^2+X+1$, nor $X^2-X+1$
	then the same conclusion holds.
\end{lemma}
\begin{proof}
	The arguments are similar to the one of \cite[Lemma~B.2]{KimSemiclassical}.
	Let $\lambda_i$ be the roots of $\chi_{A'}$
	and $Z\coloneqq \Q[\lambda_1,\ldots,\lambda_d]$ be the splitting field of $\chi_{A'}$. One has
	\begin{align*}
		d! =[Z\colon \Q] = \prod_{i=1}^d [\Q[\lambda_1,\ldots,\lambda_i]\colon \Q[\lambda_1,\ldots,\lambda_{i-1}]
		\leq \prod_{i=1}^d  d-i+1 = d!
	\end{align*}
	by the same argument as for \eqref{eq:Galcond}.
	In particular, $\Q[\lambda_i]$ has degree $d$ over $\Q$.
	This implies that $\chi_{A'}$ is irreducible.
	It follows easily that $\chi_{A'}^\ast=\chi_{A'^{-1}}=\chi_{A'^{-T}}$
	is irreducible as well.
	
	The lemma will be proved if $\lambda_i\neq \lambda_j^{-1}$ for any $i,j$ (separability property)
	and $\lambda_i^N\neq \lambda_j^{N}$ for $i\neq j$ for any $N$ (not root of unity property).
	If $\lambda_i =\lambda_i^{-1}$ then $\lambda_i^2=1$ and
	therefore $\chi_{A'}\mid X^2-1$ by irreducibility.
	Then $Z=\Q$ contradicting $[Z:\Q]=|S_d|=d!$.
	
	Let us therefore assume that $\lambda_i^N\in \Q[\lambda_j]$ for some $j\neq i$.
	Then $\lambda_i^N\in \Q[\lambda_i]\cap\Q[\lambda_j]$.
	By the Galois correspondence
	$\Q[\lambda_i]=Z^{G_i}=\{x\in Z\mid \sigma(x)=x \,\forall \sigma\in G_i\}$
	with $G_i = \{\sigma\in S_d\mid \sigma(\lambda_i)=\lambda_i\}$.
	The field $\Q[\lambda_i]\cap \Q[\lambda_j]$ is the fixed field of 
	the subgroup generated by $G_i$ and $G_j$.
	Since $G_i\simeq S_{d-1}$ and $G_j$ contains some element not fixing $\lambda_i$
	we have $\Q[\lambda_i]\cap \Q[\lambda_j]=Z^{S_d}=\Q$.
	Hence, $\lambda_i^N\in \Q$.
	Moreover, $\lambda_i^N$ is a root of $\chi_{A^N}\in \Z[X]$
	which is monic
	forcing $\lambda_i^N\in \Z$. Indeed, if $\lambda_i^N=p/q$ (with $p$ and $q$ coprime), then $q$ divides $p^{2d}$ and thus $q=\pm 1$.
	The same holds true for $\lambda_i^{-N}$
	so that $\lambda_i^N=\pm 1$ and $\lambda_i^{2N}=1$.
	As before, as $\lambda_i$ is a root of unity,
	$\chi_{A'}$ is a cyclotomic polynomial and 
	$\operatorname{Gal}(\chi_{A'})$ is abelian.
	This is a contradiction if $d\geq 3$.
	The cyclotomic polynomials of degree $2$ are the three listed ones.
\end{proof}

\subsection{Finding examples satisfying \texorpdfstring{\eqref{eq:Galcond} and the eigenvalue condition}{(G) and the eigenvalue condition}}
We now describe a method of generating $A\in \text{Sp}(2d,\Z)$ with the required properties
for Theorem~\ref{thm:mainintro}.
By Lemma~\ref{la:fullGaloisImplieshyper}, 
\eqref{eq:Galcond} and $m(A)+l(A)\geq 2$ is sufficient.
We start with the following result:
\begin{theorem}
	[{\cite[Thm.~B.1]{KimSemiclassical}}]
	\label{thm:Galois100}
	For any $d\geq 2$ the matrices in $\operatorname{Sp}(2d,\Z)$ with \eqref{eq:Galcond} have
	density one in $\operatorname{Sp}(2d,\Z)$ (when ordered by some norm on $\operatorname{M}(2d,\R)$).
\end{theorem}
Below we will present a method to produce a matrix in $\operatorname{Sp}(2d,\Z)$ with
\eqref{eq:Galcond} and $m(A)+l(A)\geq 2$
out of a matrix of $\operatorname{Sp}(2d,\Z)$ satisfying (only) \eqref{eq:Galcond}. Since both assumptions only depend on the characteristic polynomial, 
the following result reduces the problem to finding suitable polynomials thanks to the next result.

\begin{theorem}
[{\cite[Thm.~A.1]{Kirby,Rivin08}}]
\label{thm:symplecticfromchar}
	Let $f\in \Z[X]$ be a monic palindromic polynomial. 
	Then there is $A\in \operatorname{Sp}(2d,\Z)$ with characteristic polynomial $f$.
\end{theorem}

The following two statements can be used to ensure \eqref{eq:Galcond}.
The first one is a classical method to determine the Galois group.

\begin{lemma}
[{\cite[Thm.~28.23]{Isaacs}}]
	\label{la:cycletypes}
	Let $f\in \Z[X]$ be monic irreducible over $\Q$.
	For $d_1,\ldots,d_l\in \N$ the following two statements are equivalent:
	\begin{enumerate}
		\item $\operatorname{Gal}(f)$ contains a permutation (of the zeros of $f$)
			which is a disjoint product of cycles of length $d_1,\ldots, d_l$.
		\item There is a prime $p$ such that $\overline{f} = \overline{f_1}\cdots \overline{f}_l$,
			for some irreducible $\overline{f}_i\in \F_p[X]$,
			$\overline{f}$ is separable,
			and $\deg \overline{f}_i = d_i$,
			where $\overline{f}$ denotes the reduction of $f$ modulo $p$.
	\end{enumerate}
\end{lemma}

\begin{lemma}
	[{\cite[Lemma~2]{DDS98}}]
	\label{la:generatorcycle}
	If $f\in \Z[X]$ is monic  and palindromic,
	and $\operatorname{Gal}(f)$ contains a $2$-cycle, a $4$-cycle, a $(2d-2)$-cycle, and a $2d$-cycle
	then $f$ satisfies \eqref{eq:Galcond}.
\end{lemma}

 With these lemmas at hand, we assume $d\geq 4$ so that $2d \neq 4\neq 2d-2$ (the adjustments for $d=2,3$ are obvious).
 If $ f$ is the characteristic polynomial of some matrix satisfying \eqref{eq:Galcond} (which have density one by Theorem~\ref{thm:Galois100}), then $\text{Gal}(f)$ contains  a $2$-cycle, a $4$-cycle, a $(2d-2)$-cycle and a $2d$-cycle and we can take 
 the primes $p_2$, $p_4$, $p_{2d-2}$ and $p_{2d}$ given by Lemma~\ref{la:cycletypes} and corresponding respectively to these cycles in $\text{Gal}(f)$.
We observe that any polynomial $\tilde f$ that agrees with $f$ modulo the primes
$p_2,p_4,p_{2d-2}$, and $p_{2d}$
also satisfies \eqref{eq:Galcond} by Lemma~\ref{la:cycletypes} and Lemma~\ref{la:generatorcycle}.
By the Chinese remainder theorem such $\tilde f$ is unique mod $p_2 p_4 p_{2d-2} p_{2d}$.
We will choose a suitable $\tilde f$ to ensure the eigenvalue condition.

To obtain a polynomial satisfying \eqref{eq:Galcond} and the eigenvalue condition, 
we now set $f_k\coloneqq f+ k p_2p_4p_{2d-2}p_{2d}X^d $, $k\in \Z$, which is still monic and palindromic.
As explained above, $f_k$ satisfies \eqref{eq:Galcond}.
We now define $g\in \Z[X]$ by $X^dg(X+X^{-1})=f$.
Obviously, $g_k\coloneqq g + kp_2p_4p_{2d-2}p_{2d}$ satisfies $X^dg_k(X+X^{-1})=f_k$.
Pairs of roots $\{\alpha,\alpha^{-1}\}$ of $f_k$ correspond to roots of $g_k$
via $\{\alpha,\alpha^{-1}\}\mapsto \alpha+\alpha^{-1}$.
Moreover, pairs of roots of $f_k$ on $\mathbb S^1$ correspond to roots of $g_k$ in the interval $[-2,2]$ (see \cite[Lemma~1]{KonvalinaMatachePalindrome}).
{Since $g([-2,2])$ is a bounded set,} $g_k = g + kp_2p_4p_{2d-2}p_{2d}$ has no roots in $[-2,2]$ for $|k|\gg 1$.
Hence $f_k$ has no roots in $\mathbb S^1$ for $|k|\gg 1$. If $d=2$ then $g_k$ is a quadratic polynomial
which has two real roots outside $[-2,2]$ for $k\ll -1$.
Hence $f_k$ has $4$ real roots for $k\ll -1$.

We conclude that $A\in \text{Sp}(2d,\Z)$ with characteristic polynomial
$f_k$ for $|k|\gg 1$ (resp. $k\ll -1$ if $d=2$) 
satisfies \eqref{eq:Galcond} and $m(A)+2l(A)=d$ (resp. $m(A)=2$ if $d=2$).
In both cases $m(A)+l(A)\geq 2$.
The existence of $A$ with  such a characteristic polynomial is finally provided by Theorem~\ref{thm:symplecticfromchar}.

\begin{remark}
	In the above construction, instead of starting with a characteristic polynomial satisfying \eqref{eq:Galcond}
	we could also choose the reductions of $f$ modulo the primes directly.
	More precisely, 
	let $\mathcal{P}^{\text{rec}}_{2d}(\F_{p},(i))$ be the set of monic
	palindromic polynomials of degree $2d$ over $\F_p$ 
	consisting of polynomials factoring as an irreducible polynomial of degree $i$ 
	times a product of $2d-i$ distinct linear polynomials (see \cite[Appendix~B]{KimSemiclassical}).
	By \cite[Lemma~B.4]{KimSemiclassical} 
	for any integer $1\leq k\leq n$ we have
	\[
		\#\mathcal{P}^{\text{rec}}_{2d}(\F_p;(2k))= \frac 1 {2^{d-k+1}\cdot k \cdot (d-k)!} p^d +\mathcal{O}(p^{d-1})
	\]
	so that $ \mathcal{P}^{rec}_{2d}(\F_p;(2k))$ is non-empty 
	for almost all primes $p$.

	Let us now pick distinct primes $p_2,p_4,p_{2d-2},p_{2d}$
	and $q_i\in\mathcal{P}^{\text{rec}}_{2d}(\F_{p_i};(i))$ for $i\in \{2,4,2d-2,2d\}$.
	Let $f\in \Z[X]$ be monic and palindromic such that the reduction modulo the primes $p_i$ gives $q_i$.
	We are now in the same situation 
	as if $f$ was the characteristic polynomial of a matrix in $\operatorname{Sp}(2d,\Z)$ satisfying \eqref{eq:Galcond}.
\end{remark}

\begin{exam}
	The matrix 
	\[
		A = \begin{pmatrix}
			0&0&1&0\\
			0&0&0&1\\
			-1&0&-3&3\\
			0&-1&3&4
		\end{pmatrix}
	\]
	has characteristic polynomial $\chi_A= X^4-X^3-19X^2-X+1$
	and eigenvalues $4.9059\ldots$ and $-3.850\ldots$ as well as their inverses.
	In particular, $m(A)=2$. The polynomial 
	$g=X^2-X-21$ satisfies $g(X+1/X)X^2=\chi_A$ and has two real roots $\frac{1\pm \sqrt{85}}2$ lying outside $[-2,2]$.
	$\chi_A$ is irreducible $\mod 2$
	and splits into $(X-3)(X-5)(X^2+1) \mod 7$.
Hence $|\text{Gal}(\chi_A)|=8$ by Lemma~\ref{la:cycletypes} and \ref{la:generatorcycle},
i.e.~$A$ satisfies \eqref{eq:Galcond}.
By Theorem~\ref{thm:spforirreducible} $\operatorname{Sp}(4,\Z)_A = \{\pm I\}\times \Z^2$.
	Let 
	\[
		B= \begin{pmatrix}
			1&-27&45&114\\
			-27&-62&114&311\\
			-45&-114&208&564\\
			-114&-311&564&1524
		\end{pmatrix}
		= 9A^3+29A^2-21A+3I
	\]
	Clearly, $AB=BA$ and one checks $B\in \text{Sp}(4,\Z)$
	and $B$ has characteristic polynomial $\chi_B=X^4-1671X^3+17191X^2-1671X+1$.
	We observe $\chi_A\equiv \chi_B\equiv X^4+X^3+X^2+X+1 \mod 2$ but $\chi_A\not \equiv \chi_B\equiv X^4+2X^3+6X^2+2X+1 \equiv (X-3)(X-5)(X^2+3X+1)\mod 7$
	 so they define the same cycle type $\mod 2$ and $\mod 7$.
	Diagonalizing $A$ with a matrix $S\in \operatorname{Sp}(4,\R)$ gives
	\[
		S^{-1} AS = \diag(-3.8500\ldots,4.9059\ldots, (-3.8500\ldots)^{-1},(4.9059)^{-1})
	\]
	and 
	\[
		S^{-1}BS = \diag(
		0.0975\ldots,1660.6486\ldots,(0.0975\ldots)^{-1},(1660.6486\ldots)^{-1})
	\]
	From the moduli of the first two entries it now follows easily that
	$A^kB^l=I$, $k,l\in \Z$, can only hold for $k=l=0$,
	i.e.~$\Lambda\coloneqq \{A^kB^l \mid k,l\in\Z\}$ is free abelian of rank 2.
	We conclude that $\Lambda $ has finite index in $\operatorname{Sp}(4,\Z)_A$.
	Hence, we can apply Theorem~\ref{thm:mainintro} to find $B_\varepsilon\in \Lambda$.	
\end{exam}

\appendix

\section{A brief reminder on the metaplectic group}
\label{a:metaplectic}

In this section we review the metaplectic transformations and we refer to \cite{GossonSymplecticMethods} for more details. See also~\cite[Ch.~4]{Folland89}.

The metaplectic group $\text{Mp}(2d,\R)$ is the unique double covering group of $\text{Sp}(2d,\R)$.
More precisely, $\text{Sp}(2d,\R)$ has a Cartan decomposition $\operatorname{Sp}(2d,\R) \simeq K\times \mathfrak{p}$
where $K\simeq U(d)$ and $\mathfrak{p}$ is a vector space \cite[Thm.~6.31, p.~575]{Knapp}.
This shows that the fundamental group $\pi_1(\text{Sp}(2d,\R))=\pi_1(K)$ as $\mathfrak{p}$ is a vector space.
In our case, $K\simeq U(d)$ has fundamental group $\Z$.
Let for the moment $\widetilde G$ be the universal cover of $\text{Sp}(2d,\R)$ 
with covering map $\widetilde q$.
Then $\text{Sp}(2d,\R) \simeq \widetilde G / \ker \widetilde q$ and $\ker \widetilde q\simeq \Z$ is
contained in the center of $\widetilde G$.
All other coverings of $\text{Sp}(2d,\R)$ factor through this universal covering.
This means they are given by subgroups $n\Z$, $n\in \N_0$, of $\Z$.
In particular, we have $\widetilde G \to \widetilde G/2\Z \to \widetilde G/\Z = \text{Sp}(2d,\R)$
and the middle group $\widetilde G/2\Z$ is the metaplectic group $\text{Mp}(2d,\R)$.
It is the unique connected double cover of $\text{Sp}(2d,\R)$.
Let $q$ denote the covering map.
By construction, the map $q$ is $2:1$, i.e.~there is $Z$ in the center of $\text{Mp}(2d,\R)$
such that $\ker q =\{1,Z\}$.
It is worth noting that $\text{Mp}(2d,\R)$ has no finite-dimensional faithful representation,
i.e.~it cannot be described as a group of matrices.
However, there is a faithful representation $\pi_h$ on $L^2(\R^d)$ which is unique with the 
property 
\begin{equation}
	\label{eq:weilprop}
	\forall (w,t)\in H_d,\quad \pi_h(g) T_{(w,t)} \pi_h(g)^{-1} =T_{(q(g)w,t)}
	\quad g\in \text{Mp}(2d,\R),
\end{equation}
where $H_d$ is the Heisenberg group and the representation $T$ of the Heisenberg group $H_d$ is the unique
unitary irreducible representation such that $T_{(0,t)} = e^{\frac ih t}$
(i.e.~has central character $e^{\frac ih}$)
by the Stone-von Neumann Theorem.

\begin{remark}
 Recall from~\cite[\S2]{DyatlovJ} that $T$ is the representation used to defined the Weyl quantization.
\end{remark}

The construction of $\pi_h(g)$ works as follows. For $g\in \text{Sp}(2d,\R)$, the map
$(w,t)\mapsto T_{(gw,t)}$ defines another unitary irreducible representation of $H_d$
with the same central character.
Hence they are unitarily equivalent and there exits an isometry $M_g$ of $L^2(\R^d)$
such that $M_g \circ T_{(w,d)} \circ M_g^{-1} = T_{(gw,d)}$.
By Schur's Lemma $M_g$ is unique up scalar.
One can check 
that $g\mapsto [M_g]$ defines a projective representation of $\text{Sp}(2d,\R)$,
i.e.~a group homomorphism $\text{Sp}(2d,\R)\to \mathcal{U}(L^2(\R^d))/(\mathbb{S}^1\cdot \text{Id})$.
This representation does not come from an ordinary representation of $\text{Sp}(2d,\R)$
but it can be lifted to a representation of $\text{Mp}(2d,\R)$,
where 'lifted' means precisely Equation~\eqref{eq:weilprop}.
One then gets 
$\pi_h(Z)=-1$, since $Z^2=1\in \operatorname{Mp}(2d,\R)$ and $\pi_h(Z)$ must be scalar by Schur's Lemma and 
$\neq 1$ since otherwise it would be a representation of $\text{Sp}(2d,\R)$.

This allows to define the metaplectic representation $\widetilde{M}_h(A)$ of $A$ as $\pi_h(\widetilde{A})$ for any $h>0$, where we picked $\widetilde A$ to be a metaplectic lift of $A$ (which is unique up to $\text{Ker}(q)=\{1,Z\}$). This defines an operator on $L^2(\R^d)$ that can can restricted to the Schwartz class. Hence, by duality, it can be extended to tempered distributions and thus to the spaces $\mathcal{H}_\mathbf{N}$ from \S\ref{s:semiclassical}. This corresponds to the operators $M_\mathbf{N}(A)$ considered in this article.

Hence, if we are given two symplectic matrices $A$ and $B$, we can choose two metaplectic lifts $\tilde{A}$ and $\tilde{B}$ in $\text{Mp}(2d,\R)$ which are unique up to $\text{Ker}(q)=\{1,Z\}$. One has then
$$
\widetilde {M}_h(A)\widetilde M_h(B)=\pi_h(\widetilde{A})\pi_h(\widetilde{B})=\pi_h(\widetilde{A}\widetilde{B})=\pm \widetilde{M}_h(AB).
$$
In particular, if $A$ and $B$ commute, one finds that $\widetilde {M}_h(A)\widetilde M_h(B)=\widetilde {M}_h(B)\widetilde M_h(A)$ or $\widetilde {M}_h(A)\widetilde M_h(B)^2=\widetilde {M}_h(B)^2\widetilde M_h(A).$ In that last case, this yields the commuting relation $\widetilde {M}_h(A)\widetilde M_h(B^2)=\widetilde {M}_h(B^2)\widetilde M_h(A).$

\bibliographystyle{alpha}
\bibliography{refs}

\hrulefill

\end{document}